\documentclass[aps,pra,superscriptaddress,reprint,onecolumn,notitlepage]{revtex4-1}
\usepackage{verbatim}
\usepackage[utf8]{inputenc}
\usepackage[LGR,T1]{fontenc}
\usepackage[english]{babel}
\usepackage{ dsfont }
\usepackage{amsmath}
\usepackage{amssymb}
\usepackage{mathrsfs}
\usepackage{amsthm}
\usepackage{ wasysym }
\usepackage{natbib}

\usepackage{makecell}

\usepackage{tikz}
\usetikzlibrary{decorations.pathmorphing}

\definecolor{darkred}{RGB}{200, 0, 0}
\definecolor{darkblue}{RGB}{0, 0, 180}
\definecolor{darkgreen}{RGB}{50, 150, 0}
\definecolor{col2}{HTML}{64A857}
\definecolor{col3}{HTML}{D1603D}
\definecolor{orange}{rgb}{1,0.5,0}
\def\width{0.6}
\def\height{0.6}
\newcommand{\textsq}[3]
{
\node[align=center] at (#1, #2) {#3};
\draw [black] (#1 + \width, #2 + \height) rectangle (#1 - \width, #2 - \height);
}

\makeatletter
\newtheorem*{rep@theorem}{\rep@title}
\newcommand{\newreptheorem}[2]{%
\newenvironment{rep#1}[1]{%
 \def\rep@title{#2 \ref{##1}}%
 \begin{rep@theorem}}%
 {\end{rep@theorem}}}
\makeatother
\newtheorem{Thm}{Theorem}
\newtheorem{Def}[Thm]{Definition}

\newtheorem{Crl}[Thm]{Corollary}
\newtheorem{lmm}[Thm]{Lemma}
\newtheorem{ptol}[Thm]{Protocol}
\newtheorem{hyp}[Thm]{Assumptions}
\newreptheorem{ptol}{Protocol}

\def\diracspacing{0.7pt}

\newcommand{\bra}[1]{\langle #1 \hspace{\diracspacing} |} % bra
\newcommand{\ket}[1]{| \hspace{\diracspacing} #1 \rangle} % ket
\newcommand{\ketbra}[2]{| \hspace{\diracspacing} #1 \rangle \langle #2 \hspace{\diracspacing} |} % ketbra with different vectors
\newcommand{\bramatket}[3]{\langle #1 \hspace{\diracspacing} | #2 | \hspace{\diracspacing} #3 \rangle} % bra-matrix-ket with different vectors
\newcommand{\bramatketq}[2]{\bramatket{#1}{#2}{#1}} % bra-matrix-ket with the same vector

\usepackage{hyperref}

\DeclareMathOperator{\tr}{tr}

\newcommand{\corefig}
{
	\draw[rounded corners, darkgreen, dashed, ultra thick] (-4.2, 4.6) -- (4.7, 4.6) -- (4.7, -2) -- (-4.2, -2) -- (-4.2, 4.6);
	\node[align=center] at (0, 4.1) {\textbf{Alice}};
	\textsq{-3.1}{0}{main\\device};
	\draw [->] (-3.1, 1.2) to (-3.1, 0.7);
	\node[align=center] at (-1.3,0.3){$\rho_B$};
	\draw [->] (-2.4,0) to (-0.35,0);
	\node[align=center] at (-3.1, 1.45) {$\theta$};
	\draw [->] (-3.1, -0.7) to (-3.1, -1.2);
	\node[align=center] at (-3.1, -1.4) {$x$};
	\draw (0.5, 0) ellipse (0.8 and 0.5);
	\node[align=center] at (0.5, 0) {switch};
	\draw [->, in=180, out=0] (1.45, 0) to (2.85, 2.5);
	\node[align=center] at (1.75, 0.40) {or};
	\node[align=center] at (1.75, -0.2) {or};
	\textsq{3.6}{2.5}{test\\device};
	\draw [->] (3.6, 3.7) to (3.6, 3.2);
	\node[align=center] at (3.6, 3.95) {$t$};
	\draw [->] (3.6, 1.8) to (3.6, 1.3);
	\node[align=center] at (3.6, 1.05) {$y$};
	\draw [->] (1.45, 0) to (6.65, 0);
	\textsq{7.4}{0}{\textbf{Bob}};
	\draw [->] (7.4, -0.7) to (7.4, -1.2);
}

\begin{document}

\title{Device-independence for two-party cryptography and position verification}

\author{Jérémy Ribeiro}
\affiliation{QuTech, Delft University of Technology, Lorentzweg 1, 2628 CJ Delft, The Netherlands}

\author{Le Phuc Thinh}
\affiliation{QuTech, Delft University of Technology, Lorentzweg 1, 2628 CJ Delft, The Netherlands}
\affiliation{Centre for Quantum Technologies, National University of Singapore, 3 Science Drive 2, Singapore 117543}

\author{Jędrzej Kaniewski}
\affiliation{QuTech, Delft University of Technology, Lorentzweg 1, 2628 CJ Delft, The Netherlands}
\affiliation{Centre for Quantum Technologies, National University of Singapore, 3 Science Drive 2, Singapore 117543}
\affiliation{Department of Mathematical Sciences, University of Copenhagen, Universitetsparken 5, 2100 Copenhagen, Denmark}

\author{Jonas Helsen}
\affiliation{QuTech, Delft University of Technology, Lorentzweg 1, 2628 CJ Delft, The Netherlands}

\author{Stephanie Wehner}
\affiliation{QuTech, Delft University of Technology, Lorentzweg 1, 2628 CJ Delft, The Netherlands}

\date{\today} 

\begin{abstract}
    Quantum communication has demonstrated its usefulness for quantum cryptography far beyond quantum key distribution. One domain is two-party cryptography, whose goal is to allow two parties who may not trust each other to solve joint tasks. Another interesting application is position-based cryptography whose goal is to use the geographical location of an entity as its only identifying credential. Unfortunately, security  of these protocols is not possible against an all powerful adversary. However, if we impose some realistic physical constraints on the adversary, there exist protocols for which security can be proven, but these so far relied on the knowledge of the quantum operations performed during the protocols. In this work we give  device-independent security proofs of two-party cryptography and Position Verification for memoryless devices under different physical constraints on the adversary. We assess the quality of the devices by observing a Bell violation and we show that security can be attained for any violation of the {Clauser-Holt-Shimony-Horne} inequality.
\end{abstract}

\maketitle

\tableofcontents 

\section{Introduction}

Quantum communication has demonstrated its usefulness for quantum cryptography far beyond quantum key distribution (QKD).
One domain is two-party cryptography (2PC), whose goal is to allow two parties Alice and Bob to solve joint tasks, while protecting an honest party against the actions of a malicious one.  Well-known examples of such tasks are oblivious transfer~\cite{Rabin81}, bit commitment, secure identification and private information retrieval.

Another interesting application is position-based cryptography (PBC) whose goal is to use the geographical location of an entity as its (only) credential. At the heart of these is the task of position-verification (PV) where a person wants to convince the (honest) verifiers that she is located at a particular location. Quantum protocols for PV that make use of quantum communication to enhance the security have been proposed~\cite{KMS, Malaney10,BCF10,TFKW13,RG15}. We will refer to such protocols as quantum position-verification (PV).\\

Unfortunately, one cannot achieve secure 2PC and PV
without making assumptions on the power of the adversary, even using
quantum communication~\cite{Lo1997,Mayers1997,BCF10}. This is in stark contrast to QKD where security against an all-powerful adversary (obeying the laws of physics) is attainable. The reason can be traced back to a key difference between the two scenarios: while in QKD Alice and Bob can cooperate to check on the actions of the eavesdropper, in 2PC they do not trust each other and need to fend for themselves.
\\

Nevertheless, due to the practical importance of 2PC one is willing to make assumptions in order to achieve security. Classically, one often relies on computational hardness assumptions such as the difficulty of factoring large numbers. However, as technology progresses the validity of such assumptions diminishes: it has been proven that factoring can be efficiently done on a quantum computer. Most significantly, an adversary can  retroactively break the security of a past execution of a cryptographic protocol~\cite{Maurer92b}. It turns out, however, that security can also be achieved from certain physical assumptions~\cite{Crepeau94,Crepeau92}. The advantage of these comes from the fact that physical assumptions only
need to hold during the course of the protocol. That is, even if the
assumption is invalidated at a later point in time, security is not compromised.\\

If we allow quantum communication, one possible physical assumption is 
the bounded quantum-storage model~\cite{Damgard2007,BRICS21886}, and more generally, the
noisy-storage model~\cite{O_F_S,Steph_1}. Here, the adversary is allowed to have an unlimited amount of classical storage, but his ability to store quantum information is limited. This is a relevant assumption
since reliable storage of quantum information is challenging. Significantly, however, security can always be achieved by sending more qubits than the storage device can handle. Specifically, if we assume that the adversary can store at most $r$ qubits, then security can be achieved
by sending $n$ qubits, where $r \leq n - O(\log n)$~\cite{O_F_S}, 
which is essentially optimal since no protocol can be
secure if $r \geq n$~\cite{M97,Lo_whyquantum}.  
The corresponding
quantum protocols require only very simple quantum states and
measurements -- and no quantum storage -- to be executed by the honest
parties, and their feasibility has been demonstrated
experimentally~\cite{ENGLWW15,NJCKW12}. 
It is known that the
noisy-storage model allows protocols for tasks such as oblivious
transfer, bit commitment, as well as position-based cryptography~\cite{KMS, BCF10, TFKW13, RG15, Malaney10}.\\

In all these security proofs, however, one assumes perfect knowledge
of the quantum devices used in the protocol. In other words, we know
precisely what measurements the devices make, or what quantum states
they prepare. Here, we present a general method to prove security for 2PC and PV, even if we only have limited knowledge of the
quantum devices. Chiefly, we assume that the quantum devices function
as black boxes, into which we can only give a classical input, and record a classical output. The classical input indicates the choice of a measurement that we would wish to perform, although we are not guaranteed that the device actually performs this measurement. The classical output can be understood as the outcome of that measurement. The classical processing itself is assumed to be trusted. This idea of imagining black box devices is
known as device-independent (DI) quantum cryptography~\cite{Mayers2004,Acin07,Barrett12,Ekert91}. There is a large body of work in DI
QKD (see e.g.~\cite{VV12,BCK12,RUV12}), but in contrast
there is hardly any work in DI 2PC. A protocol has been
proposed by Silman~\cite{Silman11} for bit commitment which does not make physical assumptions, and hence only
achieved a weak primitive. First steps towards DI PV have also been made in~\cite{Kent2011}, and for one-sided
DI QKD in~\cite{TFKW13}.
\\

Achieving DI security for 2PC and PV presents us new with challenges which require a different approach than what is known from QKD~\cite{KW16}.
\begin{enumerate}

\item In QKD Alice and Bob trust each other, while Eve is an eavesdropper
    trying to break the protocol. As in DI QKD we will assume that the devices 
    used in the protocol are made by the dishonest party.

%In QKD Alice and Bob are always honest, while Eve is always
 % trying to break the protocol. In DI QKD it is therefore natural to
  %give the power to prepare the devices to Eve. Analogously, we will
  %assume here that all the devices used in the protocol are always
  %prepared by the dishonest party.
%
%\item The central idea of DI cryptography is to use the inherently
 % quantum feature of Bell nonlocality to test the devices. In QKD this
 % is accomplished by performing a Bell test between Alice and Bob. In
 % our scenario, since Alice and Bob do not trust each other, Alice
 % cannot trust Bob to perform the Bell test correctly. We will
 % therefore use multiple devices per user as shown in
 % Fig.~\ref{protocol_scheme}. Specifically, we need one additional device
 % for the honest party to perform a local Bell test on their own. \jed{We spoke about this once and this is actually not true. You cannot fool Bell nonlocality: Alice could pass the second measurement device to Bob and require him to perform the Bell test (of course she would choose both inputs) and this doesn't affect security. We give Alice both devices to make it more convenient for her to perform the testing.}
%
\item In QKD, after Eve has prepared and given the devices -- which she might be entangled with -- to Alice 
and Bob, there is no more direct communication between them and Eve. On the contrary in two party cryptography, 
the dishonest party, who prepared the devices, will receive back \emph{quantum communication} from these 
devices. This feature leads to different security analysis between DI QKD and DI 2PC, and also requires us 
to develop new proof techniques.

%In the following section we will see that the protocol we start
%  with uses quantum communication between Alice and Bob. This means
%  that the adversary who prepared the devices will receive
%  \emph{quantum communication} coming back from the devices. This is
%  in sharp contrast to DI QKD, in which Eve prepares the devices --
%  with which she is possibly entangled -- and then Alice and Bob
%  simply push buttons on the devices to perform measurements. That is,
%  there is no quantum communication going back to Eve. This feature
%  introduces a significant difference between the security analysis of
%  DI QKD and the DI 2PC protocol considered here and
%  requires us to develop novel proof techniques.
%
\end{enumerate}

\begin{table*}[ht]
\centering
    \begin{tabular}{|c|c|c|}
    \hline
%     & \thead{ \,Device-independent two-party cryptography\\ secure against sequential attacks } & \thead{ \,Device-independent two-party cryptography\\ with memoryless devices } \\
    & \thead{~\cite{KW16}} & \thead{This paper}\\

     \hline
     Bound on $P_{\rm guess}$ (cf. \eqref{eq:P_guess}) & \makecell{$d \left(\frac{1}{2}+\frac{1}{2}\sqrt{\frac{1+\zeta}{2}} \right)^n$} & \makecell{$\big(1-o(1)\big)\cdot \sqrt{d} \left(\frac{1}{2}+\frac{1}{2}\sqrt{\frac{1+\zeta}{2}} \right)^n $} \\ \hline
 %    Language & Proof deals with $H_{\min}$ & deals with $P_{\text{cheat}}$ \\ \hline
     Adversary memory & \makecell{reduction to classical adversary} & \makecell{deals with the memory directly} \\ 
     \hline
     Jordan's Lemma & not used & \makecell{reduction of dimensionality \\ thanks to Jordan's Lemma}\\ \hline
     \makecell{Absolute effective\\  anti-commutator} & \makecell{used} & \makecell{used}\\ 
     \hline
    \end{tabular}
    \caption{Comparison of the proof techniques used in~\cite{KW16} with those of this paper. Our new work relates the security directly to the entanglement
cost of the adversary's storage channel, however, we borrow concepts on how to test our quantum devices from the earlier work. 
Security is possible whenever $P_{\rm guess} \leq 2^{- \alpha n}$ (see equation \eqref{eq:P_guess}) for some $\alpha > 0$, which depends on the dimension $d$ of the adversary's storage device as well 
as the parameter $\zeta$ estimated during the Bell test.
Our new analysis allows to prove security for a storage device that
is at least {twice as large as the one allowed by the previous results.}}
    \label{Table}
\end{table*}

In this paper, we present a general method for proving the device-independent security of two-party cryptography and position-verification. Specifically, we improve the device-independent security proof of 2PC (which implies security for PV) in the noisy-storage model (or noisy-entanglement model) given in~\cite{KW16} thanks to new techniques. We accomplish this by introducing an appropriate model for DI in these settings, and subsequently studying a general "guessing game" that can be related to both tasks. 
To obtain DI security, we perform a Bell test on a subset of the quantum systems used in the protocol. It is an appealing feature of our analysis that security can be attained for any violation of the {Clauser-Holt-Shimony-Horne (CHSH)
inequality}~\cite{CHSH69}.

A previous analysis~\cite{KW16} permitted us to prove a bound on the cheating probability proportional to the dimension of the adversary's quantum storage (see Table \ref{Table}). To do so, they first reduce the dishonest party to a classical adversary thanks to an entropy inequality. Then they use the absolute effective anti-commutator to prove some uncertainty relations and finally lower bound some min-entropy (which is equivalent to upper-bound the cheating probability).

Here we deal directly with a quantum adversary, which permits us to prove security for an adversary quantum memory that is at least twice as large as in the previous analysis. To overcome the difficulties induced by dealing directly with the adversary quantum memory we had to use different tools (see Table \ref{Table}). While the adversary can be fully general during the course of the protocol, we assume in this work that the devices he prepared earlier are memoryless, which means that the devices behave in the same manner every time they are used. {By analogy to classical random variables such devices are often referred to as i.i.d.~devices (which stands for independent and identically distributed)}.

\subsection{Weak String Erasure}

To analyse 2PC protocols, we focus on a simpler primitive called Weak String Erasure (WSE)~\cite{Steph_1}. %, which can in turn be used to implement bit commitment using classical communication.
WSE is a two-party primitive such that if Alice and Bob are honest then at the end of its execution Alice holds a random bit string $x \in \{0,1\}^n$ and Bob holds a random substring $x_{\mathcal{I}}$ of $x$ where $\mathcal{I}$ is {a random subset of $\{1, 2, \ldots, n\}$.}
%the set of indices corresponding to the bits from $x$ that are in $x_{\mathcal{I}}$.
WSE is secure for honest Bob if Alice cannot guess the set $\mathcal{I}$ better than random chance, and for honest Alice if it is hard for Bob to guess the entire Alice's string \emph{i.e.}~if the probability that $X=\tilde X$ is low, where $X$ is the random variable corresponding to Alice's output measurement and $\tilde X$ is the random variable corresponding to Bob's guess, that is
\begin{align}
&\exists \alpha>0:\ H_{\min} (X|\text{Bob}) \geq \alpha n \notag \\
\Leftrightarrow\ & \exists \alpha>0:\ P_{\text{guess}}(X|\text{Bob})=2^{-H_{\min} (X|\text{Bob})} \leq 2^{-\alpha n}. \label{eq:P_guess}
\end{align}
For a more formal security definition of ($\alpha,\epsilon$)-WSE see~\cite{Steph_1}.

One possible implementation of WSE~\cite{Steph_1} in case of {honest parties and trusted devices is as follows}. Alice prepares $n$ EPR entangled pairs, measures randomly half of all the pairs in BB84~\cite{BB84} bases $\theta \in \{0,1\}^n$ and gets $x \in \{0,1\}^n$. At the same time, she sends the other half to honest Bob who measures it in some random bases $\theta' \in \{0,1\}^n$ and gets $z \in  \{0,1\}^n$. As Bob does not know $\theta$, he has measured some of his states in the wrong basis, so the outcome bits corresponding to these measurements provide no information about Alice's outcome. At this stage, Bob does not know which of his measurement were done in the good basis and which were done in the wrong one. After Alice and Bob have waited for a duration $\Delta t$, Alice sends $\theta$ to Bob. Bob can now compare $\theta$ with $\theta'$ and deduce the set $\mathcal{I}:=\{k \in \{0,\ldots,n\}: \theta_k= \theta'_k\}$ of indexes where Bob's bases are the same as Alice's ones. For these indexes we have $z_k=x_k$ and Bob erases all the other bits. At this stage Bob holds $(\mathcal{I},x_{\mathcal{I}})$, where $x_\mathcal{I}$ is the substring of $x$ corresponding to the set $\mathcal{I}$.
%where Alice and Bob measured in the same bases.\\

In the device-independent version of the protocol Alice holds two devices: the main device and the testing device. Alice uses the main device to prepare and measure states, and the testing device to measure states. In the honest scenario, Alice first tests her devices by proceeding to a Bell test following Protocol \ref{Test_ptol} (in section \ref{generality}), i.e. Alice checks that the states produced and measurements performed by the main device can be used to violate the CHSH inequality. Then Alice and Bob proceed as in the trusted device protocol.

It is easy to check that WSE is secure for honest Bob, since {the set $\mathcal{I}$ is determined by $\theta \oplus \theta'$ ($\oplus$ is the bitwise addition modulo two) and $\theta'$ is chosen uniformly at random by Bob. As he does not give Alice any information about it, the probability of Alice successfully guessing} $\mathcal{I}$ is $2^{-n}$.

In the dishonest Bob scenario, we can assume that it is Bob who created Alice's devices to gain extra information and compromise Alice's security. Consequently, at the very beginning of the protocol, Alice needs to test her devices. She then uses the device $n$ times to produce a bipartite state $\rho_{AB}=\sigma_{AB}^{\otimes n}$ (i.i.d.~assumption), where $\sigma_{AB}$ is an unknown but fixed state, measures the $\rho_{A}$ part to get $x\in \{0,1\}^n$ and sends the $\rho_B$ part to Bob. Bob can proceed to any kind of operation \emph{not necessarily i.i.d.} on $\rho_B$ and stores the outcome for the duration $\Delta t$ to get a cq-state $\rho_{KB'}$. When he receives $\theta$ from Alice he performs a general measurement on his cq-state which produces the guess $\tilde x$. {Bob's cheating is considered successful if} $\tilde x =x$. However, as his quantum storage is assumed to be bounded (or noisy) which impose a restriction on the possible state Bob can hold, and permits us to show that,

\begin{itemize}
  \item \textbf{WSE is secure:} for Alice against dishonest Bob who holds a noisy storage device and is allowed to create the honest party's devices (but these devices have to be memoryless), and for Bob against dishonest Alice. 
\end{itemize}

The detailed Weak String Erasure protocol is presented in section \ref{subsec:WSE} (Protocol \ref{ptol:WSE}). The precise formal result is presented in the Lemma~\ref{NE-WSE_Lemma} in section~\ref{subsec:WSE}. This result is derived from the technical result informally presented in the method section \ref{sec:Methods}.

\subsection{Position Verification}

PV has three protagonists in the honest scenario, namely two verifiers $V_1$ and $V_2$ and one prover $P$. For simplicity we restrict to position verification in one spatial dimension. The prover claims to be at some geographical position, and the PV protocol permits to prove whether this is true. The protocol is then secure if the probability that one or more dishonest provers impersonate a prover in the claimed position decays exponentially with the number of qubits {exchanged in the protocol}. \\

When the devices are trusted and the prover is honest, we can implement the protocol as follows. $V_1$ prepares $n$ EPR entangled pairs, measures half of all the pairs in some bases $\theta \in \{0,1\}^n$ to get $x \in \{0,1\}^n$, and sends the other half to the prover $P$. $V_2$ sends $\theta$ to $P$; this random string can be preshared between the verifiers before the protocol begins. When the prover receives all the information, he measures the halves of the EPR pairs he received in the bases $\theta$ to get $x$ and sends it back to both verifiers. The verifiers then check whether the prover's answer is correct, and measure the time it took between the moment they sent information and the moment they receive the answer from the prover. If the answer is correct and if the prover replies within a predefined time $\Delta t$, then {the execution of the protocol is considered successful}.\\

The honest execution of the device-independent version of the protocol is the same except that $V_1$ starts with a Bell test of his main device using the testing device. Then he executes the exact same steps as described above. The detailed protocol is given as Protocol \ref{ptol:PV}  in section \ref{subsec:DIPV}.

\begin{figure}[ht]
  \center
  \begin{tikzpicture}[scale=1.25]
    \draw[<->, thick] (6,-3) node[right]{$x$} -- (-2,-3) -- (-2,4.5) node[above] {$t$};
    \foreach \x/\xtext in {-1/V_1, 2/P,5/V_2}
        \draw (\x,-3) -- (\x,-2.9) node[below=2pt]{$\xtext$};
    \draw[rounded corners] (-1.65,-2.6) rectangle (-0.35,-1.6); 
    \node[align=center] at (-1,-2.1){\small $\mathcal{\bar M}_{V_1}\otimes \mathds{1}_P$};% V_1 box
    \draw[rounded corners] (5.5,-2.6) rectangle (4.5,-1.6); % V_2 box
    \draw[->] (-1,-2.9)-- (-1,-2.75)node[right]{\small $\theta$} -- (-1,-2.6) ;
    \draw[->] (-1,-1.6) -- (-1,-1.3) node[above]{\small $x$};
    \draw[->,dashed] (-0.4,-1.6) -- (1.5,0.4)
        node[pos=0.5,left]{$\rho_P$};
    \draw[rounded corners] (1.5,0.4) rectangle (2.5,1.4);
    \node at (2,0.9){$\mathcal{M}_P$};%P box
    \draw[->] (4.5,-1.6) --(2.5,0.4)
        node[pos=0.5,right]{\small $\theta$};
    \draw[->] (1.5,1.4) -- (-0.5,3.4)
        node[pos=0.5,right]{\small $y$};
    \draw[->] (2.5,1.4) -- (4.5,3.4)
        node[pos=0.5,left]{\small $y$};
    \draw[rounded corners] (-0.5,3.4) rectangle (-1.5,4.4);
    \node[] at (-1,3.9){$\checked$};
    \draw[rounded corners] (4.5,3.4) rectangle (5.5,4.4);
    \node[] at (5,3.9){$\checked$};
    \draw[|-|](-2.5,-2.1) -- (-2.5,4)
        node[pos=0.5,left]{$\Delta t$}; % time constraints
    \end{tikzpicture}
  \caption{$V_1$ uses $\theta$ as an input to his device, which
    creates a bipartite state $\rho_{AP}$ and sends the part $\rho_P$
    to the prover $P$, and measures the other part $\rho_A$ to produce $x\in\{0,1\}^n$ as output. At the same time $V_2$ sends
    $\theta$ to $P$. When $P$ receives the state and $\theta$ he makes a
    measurement on the state and obtains $y \in \{0,1\}^m$. He sends
    $y$ to both verifiers. The verifiers check if $y=x$ (or if $y$ is
    "close enough" to $x$), and measure the time it took to get back
    an answer from $P$ }
  \label{fig_PV}
\end{figure}
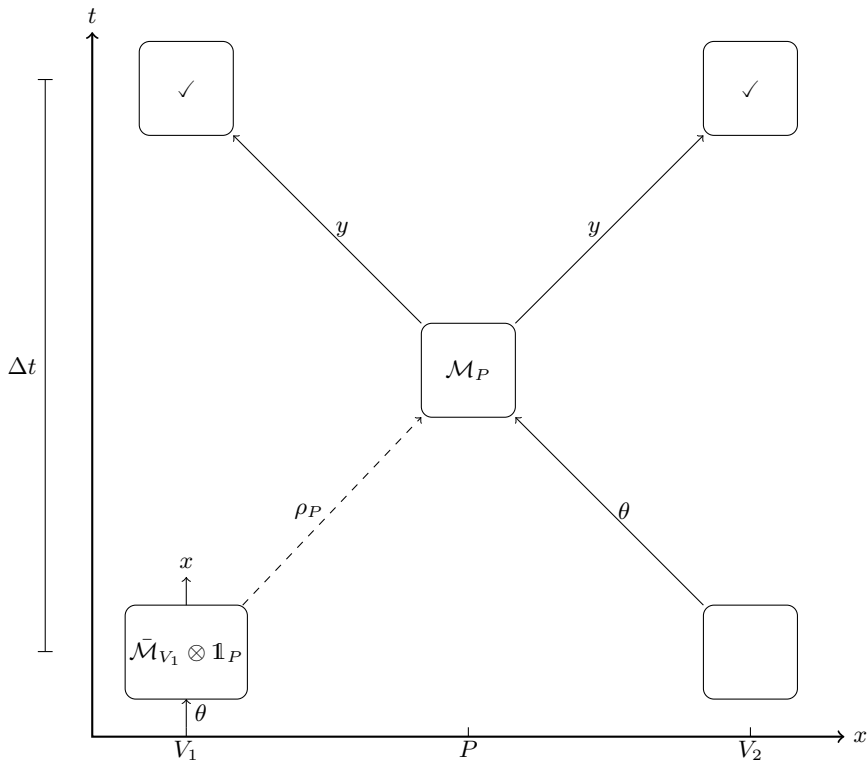

In the dishonest scenario, a single prover cannot cheat because he cannot reply on time to both verifiers. More than one dishonest prover is required and, without loss of generality, we can consider at most two dishonest provers whose goal is to impersonate one honest prover who would be at the claimed position. In this case there exists a general attack on the protocol~\cite{BCF10}. This attack, however, requires an exponential amount of entanglement with respect to the amount of quantum information received from the
verifiers. Hence, it is natural to ask if security is possible when the adversaries hold a limited amount of entanglement. We will work in this framework of cheating provers.

As security of PV can be reduced to the security of WSE, we prove that
\begin{itemize}
  \item \textbf{PV is secure}
    against adversaries who share a "noisy" entangled state and who cannot use quantum communication but
    are allowed to create the honest party's devices (these devices
    have to be memoryless).
\end{itemize}

The precise formal result is presented in Lemma \ref{bound_NE}, and this follows from a technical result informally presented in the next section.

\subsection{Methods} \label{sec:Methods}

\begin{figure}[ht]
    \includegraphics[scale=0.7]{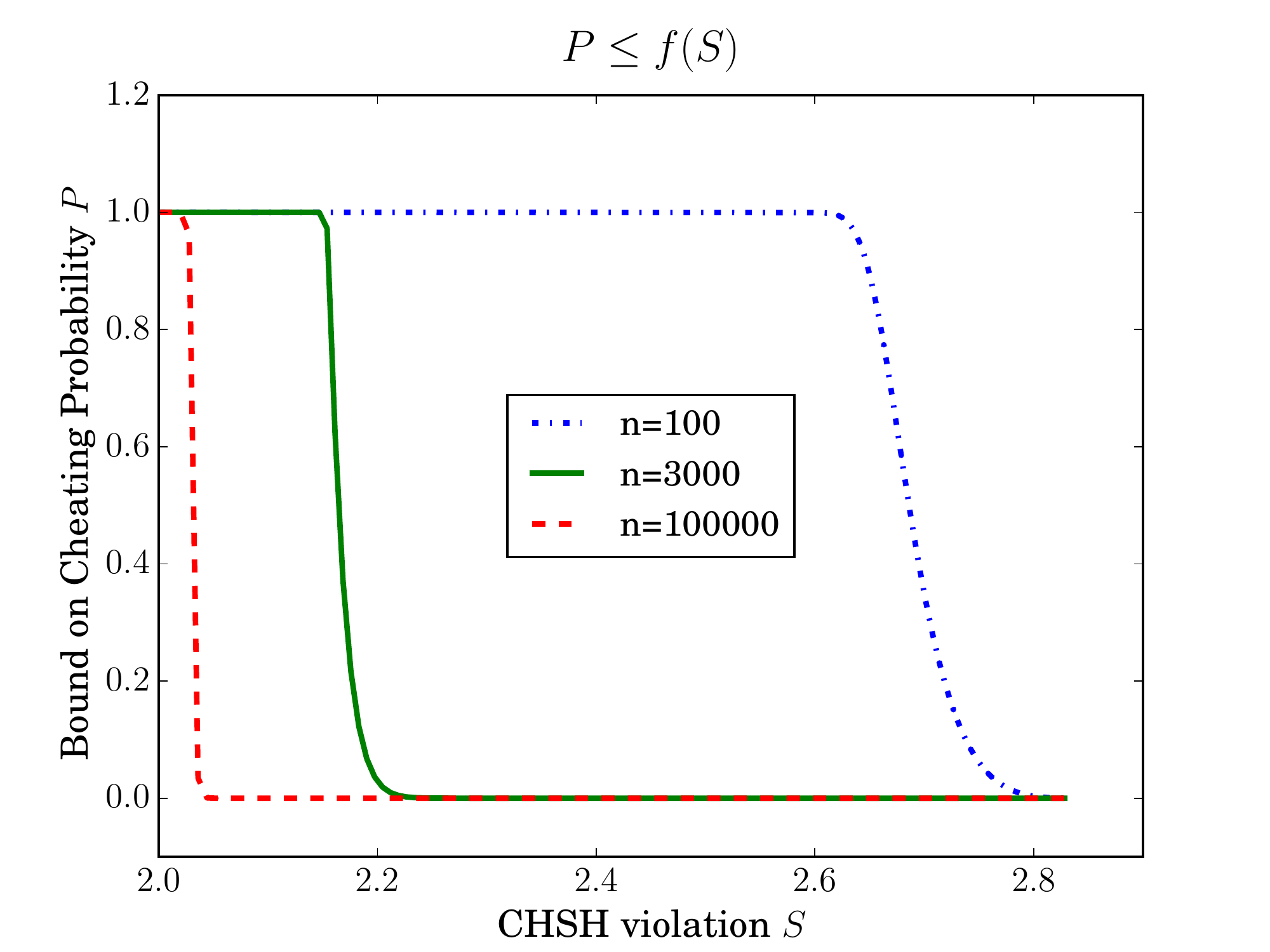}
    \caption{Security is possible for any violation of the CHSH inequality, but depending on the violation we need to send a larger number of qubits $n$.}
    \label{fig:plot}
\end{figure}

In order to prove DI security for Weak String Erasure and Position Verification, we analyse a related task known as the post-measurement guessing game. This is a two-player game where Alice plays against Bob. Alice inputs a bit string into her main device and receives an output string; Bob wins the game if he guesses correctly the output of Alice's device given his knowledge of Alice's input.\\

\begin{figure}[ht]
\begin{tikzpicture}[scale=1, line width=0.8]
	\corefig
	%\draw [->] (7.4, 1.2) to (7.4, 0.7);
	%\node[align=center] at (7.45, 1.45) {$\theta'$};
	\node[align=center] at (7.4, -1.4) {$\tilde x$};
\end{tikzpicture}
\caption{Alice is in possession of two devices prepared by Bob: The "main device" permits Alice to prepare a bipartite state $\rho_{AB}$ and measure the $\rho_A$ part of it according to a list of bases $\theta \in \{0,1\}^{n}$. The "testing device" measures according to a list of bases $t \in \{0,1\}^{m}$. These two devices are assumed to be memoryless (or i.i.d.). A dishonest Bob is assumed to be only limited by the dimension of his quantum memory, so he is allowed to make arbitrary measurements on states of dimension at most $d$. \newline \newline
 The main device prepares a bipartite state $\rho_{AB}=\sigma_{AB}^{\otimes n}$ (the tensor form follows the i.i.d.~assumption), one part $\rho_A$ is measured by the main device, {with the measurement settings specified by a random bit} string $\theta \in \{0,1\}^n$, to produce the bit string $x \in \{0,1\}^n$. The other part $\rho_B$ of the state is sent to a switch that Alice controls. As the devices are memoryless, Alice can first test her devices, and so sets her switch such that the system $B$ is sent to the testing device. She then {repeatedly performs the CHSH test} to estimate the violation. After that she sets her switch so that the system $B$ is sent to Bob. Bob's goal is to guess Alice's output $x$, i.e.~he wants to achieve $\tilde x=x$.
}
\label{protocol_scheme}
\end{figure}
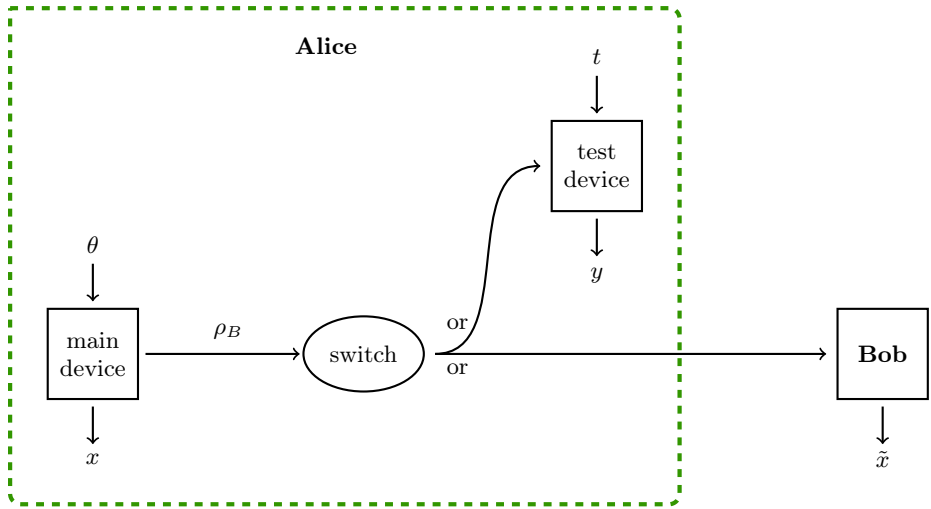

In the DI version, Alice demands that she has another test device different from her main device and dishonest Bob is allowed to create these two devices of Alice  (Fig.~\ref{protocol_scheme}). Alice can use these two devices to perform a Bell test (CHSH game), which certifies the quality of the devices. Having tested her devices, Alice uses the main device to prepare a bipartite (arbitrary) state  and measures half of it by inputting $\theta \in \{0,1\}^n$ in her main device, gets an outcome $x \in \{0,1\}^n$, and sends the other part of the quantum system to Bob. Later she sends him the input she used to perform her measurements. Once Bob has received all information he has to guess Alice's measurement outcome $x$.\\

To find a bound on Bob's winning probability, we have to assume that Bob has limited quantum storage or else he wins with certainty: he would just have to store the quantum system until he receives the bases $\theta$ and then he can measure his system in those bases. As a first step towards security against fully uncharacterized devices, we assume for now that all devices used by Alice are memoryless or i.i.d., so they behave in the same way each time Alice uses them. This implies that Alice's measurement operators are a tensor product of binary measurement operators, and the state she prepares is also of product form. This memoryless assumption also permits Alice to perform the Bell test before the actual guessing game, and from this test, to estimate an upper bound $\zeta:=\frac{S}{4}\sqrt{8-S^2}$ on a quantity we call the effective absolute anti-commutator of Alice's measurement denoted $\epsilon_+$~\cite{KW16}, where $S$ is the left hand side of the CHSH inequality. {Since $\epsilon_{+}$ is always larger than the effective anti-commutator, one can show that it gives rise to strong uncertainty relations}~\cite{KTW14}.\\

{Despite the memoryless assumption (on Alice side only), the problem remains hard}. Indeed, we cannot use techniques coming from DI QKD, since in QKD the honest parties do not send back quantum information to the eavesdropper, in contrast to the guessing game. The analysis must be different. As we do not know what Alice's measurements are, there is no limitation on the dimension on which Alice's devices act, so we cannot use bounds depending on the dimensionality of Alice's states or measurements. Moreover we have to express the absolute anti-commutator $\epsilon_+$ of Alice's measurement, in {a way that allows us to relate} it to Bob's guessing probability. In the previous work on DI WSE~\cite{KW16}, the authors reduced the problem to proving security against a classical adversary (see Table \ref{Table} for a more detailed comparison between the papers). This reduction leads to a bound which is proportional to $d$, the dimension of the adversary's quantum memory. To improve this bound we must deal with Bob's measurement, which are fully general though acting on a space of dimension at most $d$.

We overcome these difficulties thanks to Jordan's Lemma~\cite{NWZ09,Reg06}, which permits to block diagonalize Alice's measurement and reduces the dimensionality of these measurements into a list of qubit measurements. The price to pay is that we lose the "identically distributed" part of the i.i.d.~assumption on these qubit measurements. Jordan's Lemma permits us to express the absolute effective anti-commutator in an adapted form, such that we can link it to the guessing probability of Bob. Finally we prove the following.
\begin{itemize}
\item \textbf{Main technical result:} Assuming Alice's devices are memoryless, and Bob has a noisy storage device, there is a DI upper-bound on the success probability of Bob in the guessing game, which decays exponentially in $n$, the length of Alice's measurement outcomes $x \in\{0,1\}^n$ which coincides here with the number of qubits exchanged in the honest execution protocol. This bound holds for \textit{any} CHSH violation, \emph{i.e.}~$\forall S\in ]2,2\sqrt{2}]$ (see Fig.~\ref{fig:plot}).
\end{itemize}
The precise formal statement is given in Theorem \ref{main_Thm}.\\

From this result follows the DI security of WSE and PV. Indeed any attack on WSE can be viewed as a guessing game where Bob tries to guess Alice's complete string $x$. Likewise in the case of PV we can see any attack as a guessing game: the dishonest provers have to guess $V_1$'s outcome, and one can map the operations they used to the guessing game and hence show that these operations would permit Bob to win the guessing game. This implies that the cheating probability in PV is lower than that of the guessing game. This statement has been shown in ~\cite{RG15} (the remapping was done between attacks on PV and attacks on WSE, but it is essentially the same since any attack on WSE can be seen as a guessing game).

\section{Device-Independent Guessing Game}

\subsection{Preliminaries}
\label{generality}
\subsubsection{Notation}
We denote $\mathcal{H}_A$ the Hilbert space of the system $A$ with
dimension $|A|$ and $\mathcal{H}_{AB}:= \mathcal{H}_A \otimes
\mathcal{H}_B$ the Hilbert space of the composite system, with $\otimes$ the tensor product. By $\mathcal{L}(\mathcal{H})$, $\mathcal{S}_a(\mathcal{H})$, $\mathcal{P}(\mathcal{H})$ and $\mathcal{S}(\mathcal{H})$ we mean the set of linear, self-adjoint, positive semidefinite and (quantum) density operators on $\mathcal{H}$, respectively. For two operators $A,B \in
\mathcal{S}_a(\mathcal{H})$, $A \geq B$ means $(A-B) \in
\mathcal{P}(\mathcal{H})$.  If $\rho_{AB} \in \mathcal{S}(\mathcal{H}_{AB})$ then
we denote $\rho_A:=\tr_B(\rho_{AB})$ and $\rho_B:=\tr_A(\rho_{AB})$ to be the respective reduced states. A measurement, or POVM for positive operator valued measure, of dimension $d$ is a set of positive semidefinite operators that adds up to the identity operator on dimension $d$, namely
\begin{align*}
  \mathcal{F}=\bigg\{F_x, x \in \mathcal{X}: F_x \in
  \mathcal{P}(\mathcal{H}) \,\mathrm{and}\, \sum_x F_x =\mathds{1}_d\bigg\}.
\end{align*}

For $M \in \mathcal{L}(\mathcal{H})$, we denote $|M|:=\sqrt{M^\dagger M}$ and the Schatten $p$-norm $\|M\|_p:=\tr(|M|^p)^{1/p}$ for $p \in [1,\infty[$. Norms without subscript will mean the Schatten $\infty$-norm (also known as the 
operator norm): $\|M\|:=\|M\|_{\infty}$, which is the largest singular value of $M$; if $M \in \mathcal{P}(\mathcal{H})$ then $\|M\|$ is the highest eigenvalue of $M$. Some useful properties of the operator norms are $\|L\|^2=\|L^\dagger L
\|=\|L L^\dagger\|$ for all $L\in \mathcal{L}(\mathcal{H})$ and if $A,B\in \mathcal{P}(\mathcal{H})$ such that $A\geq
B$ then $\|A\| \geq \|B\|$. Moreover, whenever $A,B,L \in \mathcal{L}(\mathcal{H})$
and $A^\dagger A \geq B^\dagger B$ then $\|AL\| \geq \|BL\|$
\cite[Lemma 1]{TFKW13}.

Vector $p$-norms induce the corresponding operator $p$-norms, which we denote as $\|\cdot\|_p^{\mathrm{I}}$ to distinguish them from Schatten $p$-norms. They are defined as
\begin{equation}
    \|A\|_p^{\mathrm{I}} = \sup_{x\neq0} \frac{\|Ax\|_p}{\|x\|_p}.
\end{equation}
In the proof of a technical Lemma in the appendix, we will need the induced $1$-norm and $\infty$-norm
\begin{align}
    \|A\|_1^{\mathrm{I}} = \max_{1\leq j\leq n} \sum_{i=1}^m |a_{ij}| \quad \quad \textup{and}\quad \quad \|A\|_\infty^{\mathrm{I}} = \max_{1\leq i\leq m} \sum_{j=1}^n |a_{ij}|
\end{align} which can be seen as the maximum absolute column sum and maximum absolute row sum, respectively, and where $m$ and $n$ are the maximum row and column indexes respectively. Note that the induced $2$-norm and the operator norm are the same $\|\cdot \|_2^I=\|\cdot\|$.

For a bit string $x\in\{0,1\}^n$, $|x|$ denotes its length $n$ and the Hamming weight $w_{H}(x)$ is the number of 1's in $x$. For $x,y \in \{0,1\}^{n}$ the Hamming distance is defined as $d_{H}(x,y):= w_{H}(x \oplus y)$.

If $\mathcal{I}$ is a subset of $[n]$ then by $x_{\mathcal{I}}$ we mean the substring of $x$ with indices $\mathcal{I}$.

$E_{\textup{C,LOCC}}^{(1)}(\rho_{AB})$ is the one shot entanglement cost to create a bipartite state $\rho_{AB}$ from a maximally entangled state using only local operations and classical communication. It is formally defined as
\begin{align}
    E_{\textup{C,LOCC}}^{(1)}(\rho) :=  \min_{M,\Lambda} \{\log(M): \Lambda(\Psi^{\bar A \bar B}_M)=\rho_{AB}, \Lambda \in \textup{LOCC}, M \in \mathbb{N}\}
    \label{def_E_C},
\end{align}
where $\Psi^{\bar A \bar B}_M$ is a maximally entangled state of dimension $M$
\begin{align}
    \Psi^{\bar A \bar B}_M:= |\Psi^{\bar A \bar B}_M\rangle\langle\Psi^{\bar A \bar B}_M|, \quad \quad \quad |\Psi_M^{\bar A \bar B}\rangle := \frac{1}{\sqrt{M}} \sum_{i=1}^M |i^{\bar A}\rangle|i^{\bar B}\rangle.
\end{align}
Similarly, we have $E^{(1)}_C(\mathcal{E})$~\cite[Definition 10]{BBCW15} is the one shot entanglement cost to simulate a channel $\mathcal{E}: \mathcal{L}(\mathcal{H}_A) \rightarrow \mathcal{L}(\mathcal{H}_B)$ using LOCC and preshared entanglement:
\begin{align}
    E_C^{(1)}(\mathcal{E}):= \min_{M, \Lambda}\{\log(M): \forall \rho_A\in \mathcal{L}(\mathcal{H}_A),\ \Lambda(\rho_A \otimes \Psi^{\bar A \bar B}_M)=\mathcal{E}(\rho_A)\}
\end{align}
 where $\Lambda$ is a LOCC with $A \bar A \rightarrow 0$ (no output) on Alice's side and $\bar B \rightarrow B$ on Bob's side, and $M \in \mathbb{N}$. Note that we require a single LOCC map to simulate the effect of the channel $\mathcal{E}$ so $\Lambda$ must be independent of $\rho_A$.
 
\subsubsection{Models and assumptions}
In this section we explain in detail the assumptions imposed on the model, which are motivated by considerations on the WSE and PV protocols and our i.i.d.~constraint. 
\begin{hyp}
These are the assumptions on our device-independent guessing game:
\begin{enumerate}
  \item In device-independent protocols, the security cannot rely on the knowledge we have about the devices used by the honest party (the inner workings are unknown). These devices may even be maliciously prepared by the dishonest party to compromise security. Thus in this context, Bob is allowed to create the two devices of Alice: the main device and the testing device. These devices are assumed to be memoryless (or i.i.d.), which means that they behave in the same way every time Alice uses them. In other words, the measurements made by the devices in one round of usage depend only on Alice's input in this round (and not on previous rounds), and the state $\rho_{AB}=\sigma_{AB}^{\otimes n}$ created by her device has a tensor product form where $\sigma_{AB}$ may be chosen by Bob. The testing device is used in the testing protocol~\ref{Test_ptol}. 
   \item When Bob receives his state $\rho_{B}$ from Alice, we allow him to perform any quantum operation on it. After the operation the global state can be written as $\rho_{AB'K}$ where Alice's part $\rho_A$ has a tensor product form, and $\rho_{B'K}$ is an arbitrary qc-state held by Bob such that $|B'|\leq d$ (see assumptions \ref{Phys_assump}).
   \item Alice can test her devices before using them in the protocol as they are memoryless. We describe the testing procedure in detail in the following Protocol~\ref{Test_ptol}.
\end{enumerate}
\label{hyp}
\end{hyp}

The testing procedure aims to estimate how much the two binary measurements made by Alice's main device are incompatible given the prepared state. This is accomplished by measuring how much the main and test devices can violate the CHSH inequality.
\begin{ptol}
  Let $A_0,A_1$ be the two binary observables of Alice's main measurement device, and $T_0,T_1$ be the two binary observables of her testing device.
  \begin{enumerate}
    \item Alice creates a bipartite state $\rho_{AB}$ using her
      main device. 
    \item She sends the $B$ subsystems in state $\rho_B$ to her testing device and statistically
      estimates $S:=\tr(W \rho_{AB})$, where $W$ is the CHSH operator
      defined as
      \begin{equation}
        W:=A_0\otimes T_0 +A_0\otimes T_1 + A_1 \otimes T_0 - A_1\otimes T_1.
      \end{equation}
  \end{enumerate} \label{Test_ptol}
\end{ptol}
 The following Lemma~\ref{Bound_Eps_+} shows that this testing
 procedure permits Alice to estimate the absolute effective anti-commutator
 defined as follows.
\begin{Def}
  For two binary measurements with POVM elements $\{P_0^0,P_1^0\}$ and $\{P_0^1,P_1^1\}$, we define the absolute effective anti-commutator
  \begin{align}
    \epsilon_+:= \frac{1}{2} \tr(|\{A_0,A_1\}| \rho_A)
  \end{align}
  where $A_0:=P_0^0-P_1^0$ and $A_1:=P_0^1-P_1^1$.
  \label{Def_Eps_+}
\end{Def} 
\begin{lmm}[Proposition 2 of~\cite{KW16}]
  Let $\rho_{AT} \in \mathcal{S}(\mathcal{H}_{AT})$ and let $A_0,A_1$ and $T_0,T_1$ be binary observables on subsystem $A$ and $T$, respectively, achieving $\tr(W \rho_{AT})=:S$ for $S\geq 2$ with $W$ being the \textup{CHSH} operator. The absolute effective anti-commutator on Alice's side satisfies
  \begin{align}
    \epsilon_+ \leq \frac{S}{4} \sqrt{8-S^2}=:\zeta \in [0,1].
  \end{align}
  \label{Bound_Eps_+}
\end{lmm}

This estimation $\zeta$ of $\epsilon_+$ is central to our proof. Indeed the security bounds we derive below rely on the fact that $\zeta<1$, which means that any Bell violation in the testing procedure leads to security on WSE and PV. In other words it is enough for Alice to estimate $\zeta<1$ in the testing procedure to be sure that her devices permit her to execute the protocols (PV or WSE) securely under
\begin{hyp}
  We assume that the adversarial or dishonest party cannot have access to an unlimited and perfect quantum memory or quantum entanglement. More specifically,
  \begin{enumerate}
    \item In the guessing game and in WSE, the adversary will either have a bounded storage or a noisy storage.
    \item In PV, the adversary will either have access to bounded entanglement or noisy entanglement.
  \end{enumerate}
  \label{Phys_assump}
\end{hyp}

\subsection{Guessing games and results}
In this section, we describe and analyse the perfect and imperfect guessing games. As the name suggests, the winning condition of the perfect guessing game is more strict than that of the imperfect guessing game. Bounding the probability that Bob wins the perfect guessing game is the first step to bounding the probability that he wins the more general imperfect guessing game. The motivation behind the analysis of the imperfect guessing game is to prove security of WSE and PV even {if the protocol is made robust to noise, which is inherent to any experimental implementation}.
%Therefore we require noisy tolerant security proofs in order to implement these protocols in practice.

\subsubsection{Perfect guessing game}
We begin with a formal description of the perfect guessing game.
\begin{ptol}(Perfect guessing game)
  \begin{enumerate}
  %\item Alice tests her devices, which were created by Bob. To do that she follows the testing Protocol \ref{Test_ptol}.
  \item Alice creates $n$ identical bipartite states, chooses uniformly at random a string $\theta \in\{0,1\}^n$ and measures her $k^\mathrm{th}$ register using her main device with input $\theta_k$ to obtain an outcome $x_k$. This measurement produces an outcome string $x\in \{0,1\}^n$. At the same time she gives all the $B$ parts to Bob.
  \item Alice waits for a duration $\Delta t$ before sending her string $\theta$ to Bob.
  \item Bob tries to guess $x$ using $\theta$ and all his available information. In other words, Bob  produces an output $y$ and the (perfect) winning condition is $y=x$.
  \end{enumerate}
  \label{prot_guessing_game}
\end{ptol}

Let us analyse this game from the perspective of quantum theory and under the i.i.d.~assumption~\ref{hyp}. We will go through each step of the protocol again but with added descriptive comments. In the first step of the protocol, using the device $n$ times, Alice produces a bipartite state $\rho_{AB}=\sigma_{AB}^{\otimes n}$, and chooses the measurement setting $\theta$ to measure $\rho_A=\sigma_A^{\otimes n}$ with the POVM $\{P^\theta_x=\bigotimes_k P^{\theta_k}_{x_k}:x\in\{0,1\}^n\}$. This measurement can be seen as a tensor product of two binary measurements $\{P_0^0,P_1^0\}$ and $\{P_0^1,P_1^1\}$ because of the i.i.d.~assumption. At the same time, Alice sends to Bob a state which has i.i.d.~form $\rho_B=\sigma_B^{\otimes n}$ due to our assumption. Then, the waiting time enforces the noisy-storage model: Bob is allowed to perform any quantum operation to transform $B$ to $B'K$ where $B'$ is his quantum memory of dimension $d$ and $K$ is his (unbounded) classical memory. Bob is allowed to perform any measurement on his system $B'$, as advised by $K$ and $\theta$ and his information about the state (since he prepares the devices), in order to guess $x$. Note that for an honest implementation of the protocol, Alice does not need quantum memory, which makes the protocol easy to implement.

As the security of the protocols WSE and PV are expressed in terms of cheating probability (or equivalently in terms of min-entropy), we are here interested in the probability that Bob wins the guessing game. Indeed if this probability is low, then it means that the probability that the two protocols PV and WSE can be cheated is low also.
The winning probability of Bob in this case is nothing but the guessing probability $\lambda(n,d,\zeta)$ defined as
\begin{align}
  \lambda(n,d,\zeta):=\max_{\substack{\rho_{AB'K} \\ \mathrm{qqc}}}
  \max_{\{\mathcal{F}^\theta\}} \tr\left(2^{-n}
  \sum_{\theta,x \in \{0,1\}^n} P_x^\theta
  \otimes F_x^{\theta} \quad \rho_{AB'K}\right)
  \label{Def_lambda}
\end{align}
where the first maximization is over all qqc-states compatible with the marginal on Alice (the possible state are constraint by the value $\zeta$ Alice measured in the testing procedure), $P_x^\theta$  are the measurement operators of Alice as mentioned above and $F_x^{\theta}$ the arbitrary measurement operators of Bob acting on $B'K$ register. Note that the state $\rho_{B'K}:=\tr_A(\rho_{AB'K})$ is the qc-state that Bob gets after a quantum operation on the initial state $\rho_B=\sigma_B^{\otimes n}$ sent to him by Alice. The second maximization is a short hand for $2^n$ separate maximizations: for each $\theta$ we pick the POVM $\mathcal{F}^\theta=\{F^{\theta}_x:x\in\{0,1\}^n\}$ which maximizes the sum over $x$. The following Lemma, whose proof is presented in the appendix, gives a bound on this probability.
 \begin{lmm}[Key Lemma]
   In a perfect guessing game under the assumptions \ref{hyp} and \ref{Phys_assump} where the adversary holds a bounded quantum memory of dimension at most $d$, we have 
   \begin{align}
     \lambda(n,d,\zeta) &\leq \sqrt{d}\left( \frac{1}{2}+\frac{1}{2}\sqrt{\frac{1 + \zeta}{2}} \right)^n- \sum_{k = 0}^{t} \binom{n}{k} 2^{-n} \left( \sqrt{d}\left( \frac{1 + \zeta}{2} \right)^{k/2}-1\right) =:B(n,d,\zeta)
   \end{align}
   where $t$ is defined as
\begin{align}
t = \bigg\lfloor - \log d \cdot \bigg[ \log \bigg( \frac{1 + \zeta}{2} \bigg) \bigg]^{-1} \bigg\rfloor.
\end{align}
Observe that
\begin{align}
    B(n,d,\zeta)=\sqrt{d}\left( \frac{1}{2}+\frac{1}{2}\sqrt{\frac{1 + \zeta}{2}} \right)^n-O(n^t \cdot 2^{-n})
\end{align}
when $n\rightarrow \infty$ and $d$, $\zeta$ are constant, which decays exponentially in $n$ when $\zeta<1$.
\label{key_Lemma}
 \end{lmm}

\subsubsection{Imperfect guessing game}

The consideration of imperfect guessing game is motivated by noise in experimental realizations of any protocols. Allowing noise between provers and verifiers in WSE or PQV allows these protocols to be implemented with current state-of-the-art quantum technologies.

Formally, the imperfect guessing game consists of exactly the same steps as the guessing game discussed in the previous section, except for the winning condition of Bob. Unlike the guessing game's strict winning condition $y=x$, in the imperfect guessing game Bob wins if his guess $y$ is such that $d_H(x,y)\leq \gamma n$ for $\gamma \in [0,1[$, where $d_H(\cdot,\cdot)$ is the Hamming distance. Formally
\begin{equation}
 \lambda_{\mathrm{ip}}(n,d,\zeta,\gamma) := \max_{\substack{\rho_{AB'K}\\\mathrm{qqc}}} \max_{\{\mathcal{F}^\theta\}} \tr\bigg(2^{-n} \sum_{\theta,x \in \{0,1\}^n } \sum_{\substack{y \in \{0,1\}^n\\ d_H(x,y)\leq \gamma n}} P_x^\theta \otimes F_{y}^{\theta} \quad \rho_{AB'K}\bigg) \label{eq:pwin_impefect}
\end{equation} 

and $\gamma$ can be understood as the maximum quantum bit error rate (QBER) allowed in the protocol. We recover the perfect guessing game by taking $\gamma=0$.

One of our main results in this paper is the following
 \begin{Thm}[Main Theorem]
   For an imperfect guessing game with the maximum "QBER" allowed $\gamma\in[0,1[$, where Bob holds a noisy storage device $\mathcal{E}$ such that $E_C^{(1)}(\mathcal{E})\leq \log(d)$, the winning probability of Bob
       \begin{align}
         \lambda_{\mathrm{ip}}(n,d,\zeta,\gamma) \leq 2^{h(\gamma)n} B(n,d,\zeta)=:B'(n,d,\zeta,\gamma)
       \end{align}
       where $h(\cdot)$ is the binary entropy and $B(n,d,\zeta)$ is the bound defined in Lemma~\ref{key_Lemma}. \label{main_Thm}
 \end{Thm}
 
\begin{proof}(sketch) We first look at the imperfect guessing game in the bounded storage model where the dimension of $B'$ is bounded by $d$. 
To obtain an upper bound on $\lambda_{\mathrm{ip}}(n,d,\zeta,\gamma)$ we note that
\begin{align}
  \tr \bigg(\sum_{x \in \{0,1\}^n } \sum_{\substack{y \in \{0,1\}^n\\ d_H(x,y)\leq \gamma n}} P_x^\theta \otimes F_{y}^{\theta} \quad \rho_{AB'K} \bigg) &= \tr \bigg(\sum_{x \in \{0,1\}^n } \sum_{\substack{z \in \{0,1\}^n\\ w_H(z)\leq \gamma n}} P_x^\theta \otimes F_{x \oplus z}^{\theta} \quad \rho_{AB'K} \bigg) \\
  &= \sum_{\substack{z \in \{0,1\}^n\\ w_H(z)\leq \gamma n}} \tr\bigg(\sum_{x \in \{0,1\}^n }  P_x^\theta \otimes F_{x \oplus z}^{\theta} \quad \rho_{AB'K} \bigg).
\end{align}

\noindent Then combining the previous remark with~\eqref{eq:pwin_impefect} we have,
\begin{align}
  \lambda_{\mathrm{ip}}(n,d,\zeta,\gamma) \leq \sum_{\substack{z \in \{0,1\}^n\\ w_H(z)\leq \gamma n}} \max_{\substack{\rho_{AB'K}\\\mathrm{qqc}}} \max_{\{\mathcal{F}^\theta\}} \tr\bigg(\sum_{x \in \{0,1\}^n }  P_x^\theta \otimes F_{x \oplus z}^{\theta} \quad \rho_{AB'K} \bigg) 
\end{align}
where the first maximization is over all qqc-states compatible with the marginal on Alice.
 
\noindent Note that all the trace terms in the sum are equivalent since $z$ only permutes Bob's measurement operators. Then by using the Key Lemma~\ref{key_Lemma} to bound each term of the sum over $z$ we can write,
\begin{align}
  \lambda_{\mathrm{ip}}(n,d,\zeta,\gamma) &\leq B(n,d,\zeta) \times \sum_{\substack{z \in \{0,1\}^n\\ w_H(z)\leq \gamma n}} 1 \\
  &= B(n,d,\zeta) \times \sum_{k=0}^{\lfloor \gamma n\rfloor} \binom{n}{i}
\end{align}
% \begin{align}
%   \max_{\rho_{AB'K}} \max_{\{\mathcal{F}^\theta\}}\ \tr\left(2^{-n} \sum_{\substack{\theta \in \{0,1\}^n \\ x \in \{0,1\}^n }} P_x^\theta \otimes F_{\Gamma_q(x)}^\theta \quad \rho_{AB'K}\right) \leq B(n,d,\zeta).
% \end{align}
% Indeed (for details look the proof of Lemma \ref{key_Lemma}), the bound $B(n,d,\zeta)$ does not depend on the details of what the $F_x$ are but just on the fact that $\sum_x \|F_x\| \leq d$ which is true for all $d$ dimensional POVM $\mathcal{F}=\{F_x,\ x\in\{0,1\}^n: \sum_x F_x = \mathds{1}_d\}$, in particular it is true for $\mathcal{F}'=\{F_{\Gamma_q(x)},\ x\in \{0,1\}^n: \sum_x F_{\Gamma_q(x)} =\mathds{1}_d\}$ which is a $d$ dimensional  POVM. \newline Then since $B(n,d,\zeta)$ does not depend on $q$, we can deduce that
% \begin{align}
%   \lambda_{np}(n,d,\zeta) \leq |\mathcal{Q}| \times B(n,d,\zeta).
% \end{align}
To proceed further, we assume that $\gamma<1/2$ so $\lfloor \gamma n \rfloor$ is bounded by $\lfloor n/2 \rfloor$ and therefore by Lemma 25 of~\cite{O_F_S} we can bound the binomial sum by the binary entropy function $h(\cdot)$ so that,
\begin{align}
  \lambda_{\mathrm{ip}}(n,d,\zeta) \leq 2^{h(\gamma)n} \cdot B(n,d,\zeta)=:B'(n,d,\zeta,\gamma).
\end{align}

It remains to extend this bound to an adversary who holds a noisy memory $\mathcal{E}$ such that the one-shot entanglement cost satisfies $E_C^{(1)}(\mathcal{E}) \leq \log(d)$. Indeed, by definition of the one-shot entanglement cost~\cite{BBCW15}, the above condition means that $\mathcal{E}$ can be simulated by the identity channel $\mathds{1}_d$. Then all strategies achievable with $\mathcal{E}$ are achievable with $\mathds{1}_d$, particularly the strategy which maximizes the probability of winning in the bounded storage model. This proves the Theorem.
\end{proof}
 
The bound on the winning probability of the imperfect guessing game also decays exponentially in $n$ for suitably chosen parameters.
 \begin{lmm}
   If the maximum QBER allowed $\gamma$ satisfies the following conditions
   \begin{align}
     &\gamma \leq 1/2\\
     & h(\gamma) < -\log\left(\frac{1}{2} + \frac{1}{2}\sqrt{\frac{1+\zeta}{2}} \right)
   \end{align}
then $B'(n,d,\zeta,\gamma)$ decays exponentially in $n$, when $n \rightarrow \infty$ and $d,\zeta$ are fixed.
\label{asymp_lmm_2}
 \end{lmm}
Note that it is always possible to have a $\gamma$ which satisfies these conditions since the right hand sides of the inequalities are strictly positive.
\begin{proof} First note that $B'(n,d,\zeta,\gamma)=2^{h(\gamma) n}B(n,d,\zeta)$. According to Lemma \ref{key_Lemma}, $B(n,d,\zeta)=\sqrt{d}\left( \frac{1}{2}+\frac{1}{2}\sqrt{\frac{1 + \zeta}{2}} \right)^n-O(n^t \cdot 2^{-n})$
 It is now straightforward to see that the condition on $\gamma$
 implies the exponential decay of $B'(n,d,\zeta,\gamma)$.
 \end{proof}
 
 \section{Applications}
 \label{sec:applications}
The bound on the winning probability of the guessing game can be applied to prove the security of several two-party cryptographic protocols. Here we will apply it to prove the security of Weak String Erasure and Position Verification. For the first protocol, we can directly consider an attack on WSE as an attack on the guessing game. For the second protocol, as the security of PV can be reduced to the security of WSE~\cite{RG15}, we also get a security proof for PV.
 
\subsection{Device-Independent Weak String Erasure}
\label{subsec:WSE}
\subsubsection{({\fontencoding{LGR}\selectfont a,e=0})-WSE in the noisy storage model}

Let the two protagonists of
($\alpha,\epsilon$)-WSE be Alice and Bob. The goal of this cryptographic primitive is that at the end of its execution Alice holds a random bit string $x$ and Bob holds a
random substring of $x$ called $x_{\mathcal{I}}$. We can view this
$x_{\mathcal{I}}$ as $x$ where we have randomly erased some bits, hence the name WSE (Protocol \ref{ptol:WSE}). For a formal definition of ($\alpha,\epsilon$)-WSE we refer to
~\cite{Steph_1}.  
\begin{ptol}[Weak String Erasure] In the case where Alice and Bob are honest, the protocol is executed as follows:
\begin{enumerate}
   \item Alice tests her devices following the testing protocol~\ref{Test_ptol} and obtains $\zeta$, an estimate of an upper bound on the absolute effective anti-commutator.
  \item Alice creates $n$ identical bipartite states. She chooses uniformly at random a string $\theta\in\{0,1\}^n$ and measures her part of the $k^{\mathrm{th}}$ register by inputting it and $\theta_k$ to her measurement device to get an outcome $x_k$. This process generates an outcome string $x\in\{0,1\}^n$. At the same time she sends all the $B$ registers to Bob.
  \item Bob chooses uniformly at random $\theta'\in\{0,1\}^n$, and measures his registers in the same manner as Alice to get an outcome string $x'\in\{0,1\}^n$.
  \item Alice waits for a duration $\Delta t$ before sending $\theta$ to Bob.
  \item Bob determines the index set $\mathcal{I}:=\{k\in[n]: \theta'_k = \theta_k\}$, and obtains the corresponding substring $x'_\mathcal{I}$.
\end{enumerate}
\label{ptol:WSE}
\end{ptol}
At the end of the protocol Alice holds $x$ and Bob holds
$(\mathcal{I},x'_{\mathcal{I}})$. It can be easily checked that in the ideal implementation, $x'_{\mathcal{I}}$ is a substring of $x$ so Bob does not know the full $x$ and Alice does not know $\mathcal{I}$. WSE is secure for an honest Bob if Alice cannot get the set $\mathcal{I}$, and for an honest Alice if it is hard for Bob to guess the entire string $x$.

The probability that Bob guesses $x$, and so that he wins the game, is defined as
\begin{align}
    \lambda_\textup{WSE}(n,d,\zeta):=\max_{\rho_{AB K}}
  \max_{\{\mathcal{F}^\theta\}} \tr\left(2^{-n}
  \sum_{\theta,x \in \{0,1\}^n} P_x^\theta
  \otimes F_x^{\theta} \quad \rho_{ABK}\right)
\end{align}
where the first maximization is over all qqc-states compatible with the marginal on Alice. Saying that $x$ is hard to guess means that this probability decays exponentially with $n$ or equivalently \[\exists \alpha >0:\ H_{\min}(X|BK\Theta)/n \geq \alpha.\]

\begin{comment}
From a theoretical perspective, we can describe each step of the protocol
\begin{ptol}[Weak String Erasure]
\begin{enumerate}
   \item Alice test her device following the testing protocol~\ref{Test_ptol} and obtain $\zeta$ (with some confidence).
  \item Alice creates a bipartite state $\rho_{AB}$ such that $\rho_{AB}= \sigma_{AB}^{\otimes n}$.
  \item Alice chooses uniformly at random a string $\theta\in\{0,1\}^n$ as input to her main measurement device. The device
  makes a measurement $\mathcal{M}^\theta = \{P_x^\theta =
  \bigotimes_{j\in [n]} P_{x_j}^{\theta_j}: P_{x_j} \in
  \mathcal{P}(\mathcal{H}_A)\}$ on her state
  $\rho_{A}=\sigma_A^{\otimes n}$. From this measurement she get $x\in
  \{0,1\}^n$.
  \item Honest Bob chooses uniformly at random $\tilde \theta \in
    \{0,1\}^n$, an makes the measurement $\mathcal{M}^{\tilde
      \theta}=\{P_x^{\tilde \theta} = \bigotimes_{j\in [n]}
    P_{x_j}^{\tilde \theta_j}: P_{x_j} \in
    \mathcal{P}(\mathcal{H}_A)\}$ and gets $\tilde{x} \in \{0,1\}^n$.
  \item Alice and Bob wait during a certain amount of time $\Delta t$
  \item Alice sends her string $\theta$ to Bob.
  \item Knowing $\theta$ Bob can know which of his measurements were
    done with the right $\theta_j$, and he can extract from $\tilde{x}$,
    $\mathcal{I} \subset [n]$ and $x_{\mathcal{I}}$, where
    $\mathcal{I}$ is a set of index $\mathcal{I}:=\{j \in [n], \tilde
    x_j = x_j\}$, and $x_\mathcal{I}$ is a substring of $x$.
\end{enumerate}
\end{ptol}
\end{comment}

If Bob is dishonest, we can look at any attack strategy of Bob as a guessing strategy in the guessing game where Bob has to guess Alice's bit string $x$. Thus we have the following:
\begin{Crl}
  For $(\alpha,\epsilon=0)$-WSE in the noisy storage model, under the
  assumption~\ref{hyp}, if Alice's memoryless device is such that $\zeta<1$ then the cheating probability $\lambda_{\mathrm{WSE}}(n,d,\zeta)$ of Bob is upper bounded by $B'(n,d,\zeta,\gamma)$, where $B'(n,d,\zeta,\gamma)$ is defined in Theorem~\ref{main_Thm}.
\end{Crl}
\begin{proof}We can directly apply Theorem~\ref{main_Thm}
on ($\alpha,\epsilon=0$)-WSE by considering Bob's cheating strategy
as a guessing game.
\end{proof}
\subsubsection{({\fontencoding{LGR}\selectfont a,e=0})-WSE in noisy entanglement model}

In order to make the link between WSE and PV, we describe briefly WSE in the noisy entanglement model (see~\cite{RG15} for more details). The protocol is the same as before but now there are two Bobs, called
Bob1 and Bob2, who share an entangled state $\rho_{B_1 B_2}$ such that $E_C^{(1)}(\rho_{B_1B_2}) \leq \log(d)$ (which replaces Bob's channel $\mathcal{E}$ used in the noisy storage model), and can only communicate classically from Bob1 to Bob2. It is Bob2 who is asked to get the pair $(\mathcal{I},x_{\mathcal{I}})$, while Alice sends $\rho_{B}$ to Bob1 and $\theta$ to Bob2. If the Bobs are cheaters, Bob1 will try to send $\rho_{B}$ to Bob2 using their entanglement and classical communication, in order to enable Bob2 to guess the full outcome string $x\in \{0,1\}^n$ of Alice in the perfect case (or at least $(1-\gamma)n$ bits in the imperfect case). 

The Bobs play the role of the malicious provers in PV, called $M_1$ and $M_2$ who both want to guess $x$. The fact that in PV they both have to guess $x$ to be able to cheat the protocol makes PV harder to cheat than WSE in the noisy-entangled model where only one Bob (Bob2) needs to guess $x$. Because it is harder to cheat in PV, proving the security on this model of WSE proves the result for PV~\cite{RG15}. Again we say that WSE in the noisy-entangled model is secure if the cheating probability denoted by $\lambda_{NE}$ decays exponentially with $n$. In the two following Lemmas we first prove the security of WSE for the bounded-entanglement model, and then extend it to noisy-entanglement model.
\begin{Def}
    For $(\alpha,\epsilon=0)$-WSE in the bounded-entanglement model, the probability $\lambda_{BE}(n,d,\zeta)$ that Bob2 perfectly guesses Alice's output string $x \in \{0,1\}^n$ is,
    \begin{align}
        \lambda_{\mathrm{BE}}(n,d,\zeta):=\max_{\rho_{AB_2 K}}
  \max_{\{\mathcal{F}^\theta\}} \tr\left(2^{-n}
  \sum_{\theta,x \in \{0,1\}^n} P_x^\theta
  \otimes F_x^{\theta} \quad \rho_{AB_2K}\right)
    \end{align}
where the first maximization is over all qqc-states compatible with the marginal on Alice (which are constraint by the value $\zeta$ measured in the testing procedure). Here the state $\rho_{B_2}$ is of dimension at most $d$. 
\end{Def}

In the following Lemma~\ref{BE_Lemma} we look at the special case where
$\rho_{B_1B_2}$ is a maximally entangled state of local dimension $d$ (this case is WSE in the bounded
entanglement model). This Lemma is a variant of
Lemma~\ref{key_Lemma}.
\begin{lmm}
  For WSE in the bounded-entanglement model, where the two Bobs share a perfect entangled state $\rho_{B_1B_2}$ of dimension at most $d^2$, the probability $\lambda_{BE}(n,d,\zeta)$ that Bob2 perfectly guesses Alice's output string $x \in \{0,1\}^n$ is
  \begin{align}
    \lambda_{\mathrm{BE}}(n,d,\zeta) \leq B(n,d,\zeta)
  \end{align}
  where $B(n,d,\zeta)$ is defined in Lemma~\ref{key_Lemma}.
  \label{BE_Lemma}
\end{lmm}
\begin{proof}
The guessing probability of Bob2 in this model is given by
\begin{align}
  \lambda_{\mathrm{BE}}(n,d,\zeta):=\max_{\rho_{AB_2 K}}
  \max_{\{\mathcal{F}^\theta\}} \tr\left(2^{-n}
  \sum_{\theta,x \in \{0,1\}^n} P_x^\theta
  \otimes F_x^{\theta} \quad \rho_{AB_2K}\right)
\end{align}
where the first maximization is over all qqc-states compatible with the marginal on Alice.
Note that this is the same expression as $\lambda(n,d,\zeta)$ except that the state $\rho_{AB'K}$ is replaced by $\rho_{AB_2K}$. We can then invoke Lemma~\ref{key_Lemma} since Bob2's measurements are also
acting jointly on $d$-dimensional quantum register $B_2$ and an arbitrary large classical register $K$. 
\end{proof}

We now want to extend the result to the case where the adversary holds noisy entanglement and must guess Alice's string up to some error tolerance $\gamma$.
\begin{Def}
    For $(\alpha,\epsilon=0)$-WSE in the noisy-entanglement model, the probability $\lambda_{NE}(n,d,\zeta,\gamma)$ that Bob2 guesses Alice's output string $x \in \{0,1\}^n$ with an error rate, as defined in the paragraph before equation \eqref{eq:pwin_impefect}, at most $\gamma$ is,
    \begin{align}
        \lambda_{\mathrm{NE}}(n,d,\zeta,\gamma) &:=\max_{\rho_{AB_2K}} \max_{\{\mathcal{F}^\theta\}}\ \tr\left(2^{-n} \sum_{\theta,x \in \{0,1\}^n} \sum_{\substack{y \in \{0,1\}^n \\ d_H(x,y)\leq \gamma n}}P_x^\theta \otimes F_{y}^\theta \quad \rho_{AB_2K}\right)
    \end{align}
where the first maximization is over all qqc-states compatible with the marginal on Alice. Here we assume that the state shared by the two Bobs $\rho_{B_1B_2}$ is such that $E_C^{(1)}(\rho_{B_1B_2}) \leq \log(d)$.
\end{Def}

Now we tackle the general case where $\rho_{B_1B_2}$ is a noisy-entangled state such that $E_C^{(1)}(\rho_{B_1B_2}) \leq \log(d)$.
\begin{lmm}
  Consider $(\alpha,\epsilon=0)$-WSE in the noisy-entanglement model, where the two Bobs share a noisy entangled state $\rho_{B_1 B_2}$ such that $E_C^{(1)}(\rho_{B_1 B_2}) \leq \log(d)$. If Alice's device is such that $\zeta <1$, then the probability $\lambda_{\mathrm{NE}}(n,d,\zeta)$ that Bob2
  produces a guess $y\in\{0,1\}^n$ and
  $d_H(x,y) \leq \gamma n$ with $x \in\{0,1\}^n$ being Alice's output string, is upper bounded
  as follows
  \begin{align}
    \lambda_{\mathrm{NE}}(n,d,\zeta,\gamma)\leq B'(n,d,\zeta,\gamma)
  \end{align}
  where $B'(n,d,\zeta,\gamma)$ is defined in Theorem~\ref{main_Thm}. \label{NE-WSE_Lemma}
\end{lmm}
\begin{proof} We first look at the imperfect guessing game in the bounded entanglement model, where Bob1 and Bob2 share a maximally entangled state of dimension $M\leq d$:
%, where $d$ is the local dimension of this state:
\begin{equation}
  \left|\Psi_M^{B_1B_2} \right>:= \frac{1}{\sqrt{M}} \sum_{i=1}^M \left|i^{B_1}\right>\left|i^{B_2}\right>
\end{equation}
Denote $\Psi_M:=\left|\Psi_M^{B_1B_2}\right>\left<\Psi_M^{B_1B_2}\right|$. Note that the fact that the local dimension $\Psi_M$ is at most $d$ implies that Bob2's quantum state $\rho_{B_2}$ has a dimension bounded by $d$. Hence it is easy to see that
 \begin{align}
   \lambda_{\mathrm{NE}}(n,d,\zeta,\gamma) &:=\max_{\rho_{AB_2K}} \max_{\{\mathcal{F}^\theta\}}\ \tr\left(2^{-n} \sum_{\theta,x \in \{0,1\}^n } \sum_{\substack{y \in \{0,1\}^n \\ d_H(x,y)\leq \gamma n}} P_x^\theta \otimes F_{y}^\theta \quad \rho_{AB_2K}\right) \\
   & \leq \max_{\rho_{AB_2K}} \max_{\{\mathcal{F}^\theta\}}\ \tr\left(2^{-n} \sum_{\theta,x \in \{0,1\}^n } \sum_{\substack{z \in \{0,1\}^n \\ w_H(z)\leq \gamma n}} P_x^\theta \otimes F_{x \oplus z}^\theta \quad \rho_{AB_2K}\right)\\
   & \leq \sum_{\substack{z \in \{0,1\}^n \\ w_H(z)\leq \gamma n}} \max_{\rho_{AB_2K}} \max_{\{\mathcal{F}^\theta\}}\ \tr\left(2^{-n} \sum_{\theta,x \in \{0,1\}^n }  P_x^\theta \otimes F_{x \oplus z}^\theta \quad \rho_{AB_2K}\right)
%   &\leq \sum_q \max_{\rho_{AB_2K}} \max_{\{\mathcal{F}^\theta\}}\ \tr\left(2^{-n} \sum_{\theta,x \in \{0,1\}^n} P_x^\theta \otimes F_{\Gamma_q(x)}^\theta \quad \rho_{AB_2K}\right)
 \end{align}
where the first maximization is over all qqc-states compatible with the marginal on Alice,
can be bounded by the techniques in the proof of the Theorem~\ref{main_Thm} since the register $B_2$ has bounded dimension. We have
%where $\rho_{AB_2K}$ is the the state shared by Alice and Bob2 (we trace the global state on Bob1). Bob2 hols a qc-state $\rho_{B_2K}$ where the dimension of $\rho_{B_2}$ is bounded by $d$. This bound on the dimension of $\rho_{B_2}$ will imply that Bob2's measurement are $d$-dimensional measurements. Then we can bound
% \begin{align}
%   \max_{\rho_{AB_2K}} \max_{\{\mathcal{F}^\theta\}}\ \tr\left(2^{-n} \sum_{\theta,x \in \{0,1\}^n} P_x^\theta \otimes F_{\Gamma_q(x)}^\theta \quad \rho_{AB_2K}\right) \leq B(n,d,\zeta).
% \end{align}
% Indeed (for details look the proof of Lemma \ref{key_Lemma}), the bound $B(n,d,\zeta)$ do not depend on the details of what are the $F_x$ but just on the fact that $\sum_x \|F_x\| \leq d$ which is true for all $d$ dimensional POVM $\mathcal{F}=\{F_x,\ x\in\{0,1\}^n: \sum_x F_x = \mathds{1}_d\}$, in particular it is true for $\mathcal{F}'=\{F_{\Gamma_q(x)},\ x\in \{0,1\}^n: \sum_x F_{\Gamma_q(x)} =\mathds{1}_d\}$ which is a $d$ dimensional POVM. \newline Then since $B(n,d,\zeta)$ do not depend on $q$, we can deduce that
 \begin{align}
   \lambda_{\mathrm{NE}}(n,d,\zeta,\gamma) \leq 2^{h(\gamma)n } \times B(n,d,\zeta) \leq B'(n,d,\zeta,\gamma).
 \end{align}
% Now by saying that $|\mathcal{Q}|\leq 2^{2h(\gamma) n}$ we get the bound $B'(n,d,\zeta,\gamma)$.\newline 

We can extend this bound against an adversary who holds a noisy entangled state $\rho_{B_1B_2}$ such that $E_{\textup{C,LOCC}}^{(1)}(\rho_{B_1B_2}) \leq \log(d)$. Indeed by definition (eq. \eqref{def_E_C}) of the one shot entanglement cost of the state $\rho_{B_1B_2}$ denoted $E_{\textup{C,LOCC}}^{(1)}(\rho_{B_1B_2})$ ~\cite{BD11}, saying that $E_{\textup{C,LOCC}}^{(1)}(\rho_{B_1B_2})\leq \log(d)$ means that $\rho_{B_1B_2}$ can be created from a perfectly entangled state $\Psi_M$ of dimension $M\leq d$. Thus, all strategies achievable with $\rho_{B_1B_2}$ are achievable with $\Psi_M$. In particular the strategy which maximizes the probability of winning with respect to $\rho_{B_1B_2}$ is achievable with $\Psi_M$ which proves the Lemma.
 \end{proof}

If $\gamma$ in the protocol is such that it satisfies the condition of Lemma~\ref{asymp_lmm_2}, the previous bound proves the security of ($\alpha,\epsilon=0$)-WSE since it decays exponentially.

\subsection{Device-Independent Position Verification}
\label{subsec:DIPV}
In the following we will prove that PV in the noisy entanglement model (NE) is device-independently
secure. Indeed the attacks on PV in the NE model can be mapped to attacks on WSE in the NE model~\cite[Theorem 14]{RG15}. As we have proved in Lemma~\ref{NE-WSE_Lemma} that WSE in the NE model is device-independently secure, PV in the NE model must be secure.

Here we only speak about the one dimensional position verification protocol. In PV there are three protagonists in the honest case: two verifiers ($V_1$ and $V_2$) and one prover ($P$). The prover claims to be at some geographical position, and the PV protocol permits to check whether this is true.

\begin{ptol}[Position Verification] Let us assume $P$ has claimed his position to be in the middle of both verifiers (fig. \ref{fig_PV}). The verifiers check this claim by the following procedure:
  \begin{enumerate}
    \item $V_1$ tests his devices as described in the testing protocol
      \ref{Test_ptol}
    \item At the beginning of the protocol, the two verifiers $V_1$
      and $V_2$ share a random bit string $\theta \in \{0,1\}^n$.
    \item $V_1$ prepares a bipartite state (which is ideally a maximally entangled state)           $\rho_{V_1 P}=\sigma_{V_1 P}^{\otimes n}$ which has a tensor product
      structure, and sends the part $\rho_P:=\tr_{V_1}(\rho_{V_1 P})$
      to the prover. 
      
      $V_2$ sends the string $\theta$ to the prover, such that the
      prover receives $\theta$ and $\rho_P$ at the same time.
      
      $V_1$ applies the measurement $\mathcal{M}^\theta=\{P_y^\theta:=\bigotimes_{j \in [n]}
      P_{y_j}^{\theta_j}, y\in \{0,1\}^n \}$ to his
      part of the state $\rho_{V_1 P}$ and gets $x$.
    \item The prover applies a projective measurement
      $\mathcal{M}^\theta$, and gets a bit string $y
      \in \{0,1\}^n$. Then he sends $y$ to both verifiers.
    \item  Then $V_1$ compares
      (using the Hamming distance) his outcome $x$ with the string
      $y$ he receives from the prover, and measures how much time passed
      between the moment he sent the state to $P$ and the moment he
      receives $y$ from $P$.
    \item $V_1$ sends $x$ to $V_2$ so $V_2$ can also compare $y$ and
      $x$. $V_2$ also measures how long it took for the message to come back.
    \item If $x=y$ (or $d_H(x,y)\leq \gamma n$) and the time measured by the verifiers is lower
      than a certain fixed bound $\Delta t$ then the prover passes the
      protocol, which means that the verifiers accept that the prover is at his claimed position.
  \end{enumerate}
  \label{ptol:PV}
\end{ptol}

As mentioned in the introduction, if the prover is dishonest, it suffices to consider the scenario where there are two dishonest provers $B_1$ and $B_2$ who impersonate being at some claimed location. The protocol is secure against adversaries holding an entangled state $\rho_{B_1B_2}$ with one shot entanglement cost bounded by $d$ if the probability that the adversaries cheat the protocol decays exponentially with the length $n$ of $x$ (which is also the number of quantum system the verifiers send to the prover).

\begin{Def}
    In the general case when the winning condition on the prover's guess $y$ is $d_H(x,y)\leq \gamma n$, where $x$ is the verifiers' bit string and $\gamma \in [0,1[$ is the maximal QBER, the probability of cheating in PV is defined as
    \begin{align*}
        \lambda_{\textup{PV}}(n,d,\zeta,\gamma):= \max_{\rho_{V_1 M_1 M_2}} \max_{\{\mathcal{T}^\theta\}} \max_{\{\mathcal{F}^\theta\}} \tr\left(2^{-n}\sum_{\theta,x \in \{0,1\}^n} \sum_{\substack{y \in \{0,1\}^n \\ d_H(x,y)\leq \gamma n}} P_x^\theta \otimes T_y^\theta \otimes F_y^\theta \rho_{V_1 M_1 M_2}\right)
    \end{align*}
    where the first maximization is over all states compatible with the marginal on $V_1$, $P_x^\theta$, $T_x^\theta$ and $F_x^\theta$ are the measurement operators for $V_1$, $M_1$ and $M_2$ respectively, and where the second and the third maximisations are short hand for $2^n$ separate maximisations: for each $\theta$, $M_1$ and $M_2$ choose the POVMs which maximize $\lambda_{\textup{PV}}(n,d,\zeta,\gamma)$
\end{Def}

\begin{Def}
  PV is said to be $\alpha$-secure if there exists $\alpha>0$ and an integer
  $N\geq 1$ such that $\forall n\geq N$ the probability
  $\lambda_{\mathrm{PV}}(n,d,\zeta,\gamma)$ that dishonest provers pass the protocol is
  such that:
  \begin{align}
    \lambda_{\mathrm{\mathrm{PV}}}(n,d,\zeta,\gamma) \leq 2^{-\alpha n}.
  \end{align}
  Note that the value of $\alpha$ may depend on $d,\zeta$ and $\gamma$.
\end{Def}

In our case we limit the attack scheme by assuming that
the adversaries can only share a limited amount of entanglement
(assumptions \ref{Phys_assump}) and that they do \emph{not} use quantum
communication, but they have access to perfect and unlimited classical
communication. Moreover we will assume that the device-independent
assumption \ref{hyp} is satisfied in our model of attack.

\begin{lmm}
  In PV in the Noisy Entanglement model, where Bob1 and Bob2 share a
  state $\rho_{B_1 B_2}$ such that $E_{\textup{C,LOCC}}^{(1)}(\rho_{B_1 B_2}) \leq
  \log(d)$, if $V_1$'s device is such that $\zeta <1$ then the
  probability $\lambda_{\mathrm{PV}}(n,d,\zeta,\gamma)$ that Bob2 guesses a string $y \in
  \{0,1\}^n$ and $d_H(x,y) \leq \gamma n$,
  where $x$ is $V_1$'s outcome measurement, is upper bounded by
  \begin{align}
    \lambda_{\mathrm{PV}}(n,d,\zeta,\gamma) \leq B'(n,d,\zeta,\gamma)
  \end{align}
  where $B'(n,d,\zeta,\gamma)$ is defined in Theorem \ref{main_Thm}.
  \label{bound_NE}
\end{lmm}
\begin{proof} We use the proof in \cite[Theorem 14]{RG15},
which reduces the security of PV under the assumption that there is no quantum communication between cheaters, to the security of Weak String
Erasure in the noisy entanglement model, in other words it proves that $\lambda_{\mathrm{PV}}(n,d,\zeta,\gamma)\leq
\lambda_{NE}(n,d,\zeta,\gamma)$ and then using Lemma \ref{NE-WSE_Lemma} we conclude
the proof. 
\end{proof}

If $\gamma$ is such that it satisfies the condition of Lemma
\ref{asymp_lmm_2} this bound proves the security of PV since $B'(n,d,\zeta,\gamma)$ decays exponentially (see figure \ref{fig:noise_tolerance}). The security proof is independent of the implementation of the protocol. Moreover, to allow an honest prover to pass the protocol even when there is some noise in the quantum channel between $V_1$ and $P$ or if honest prover's measurements are not perfect means that we allow the prover $P$ to guess the string $x$ with some error quantified by the Hamming distance. This choice obviously makes the protocol easier to cheat on when $P$ is dishonest, but according to Lemmas \ref{bound_NE} and \ref{asymp_lmm_2} the protocol is still secure if the fraction of errors $\gamma$ allowed in the guessed string is small enough.

\begin{figure}[ht]
    \includegraphics[scale=0.7]{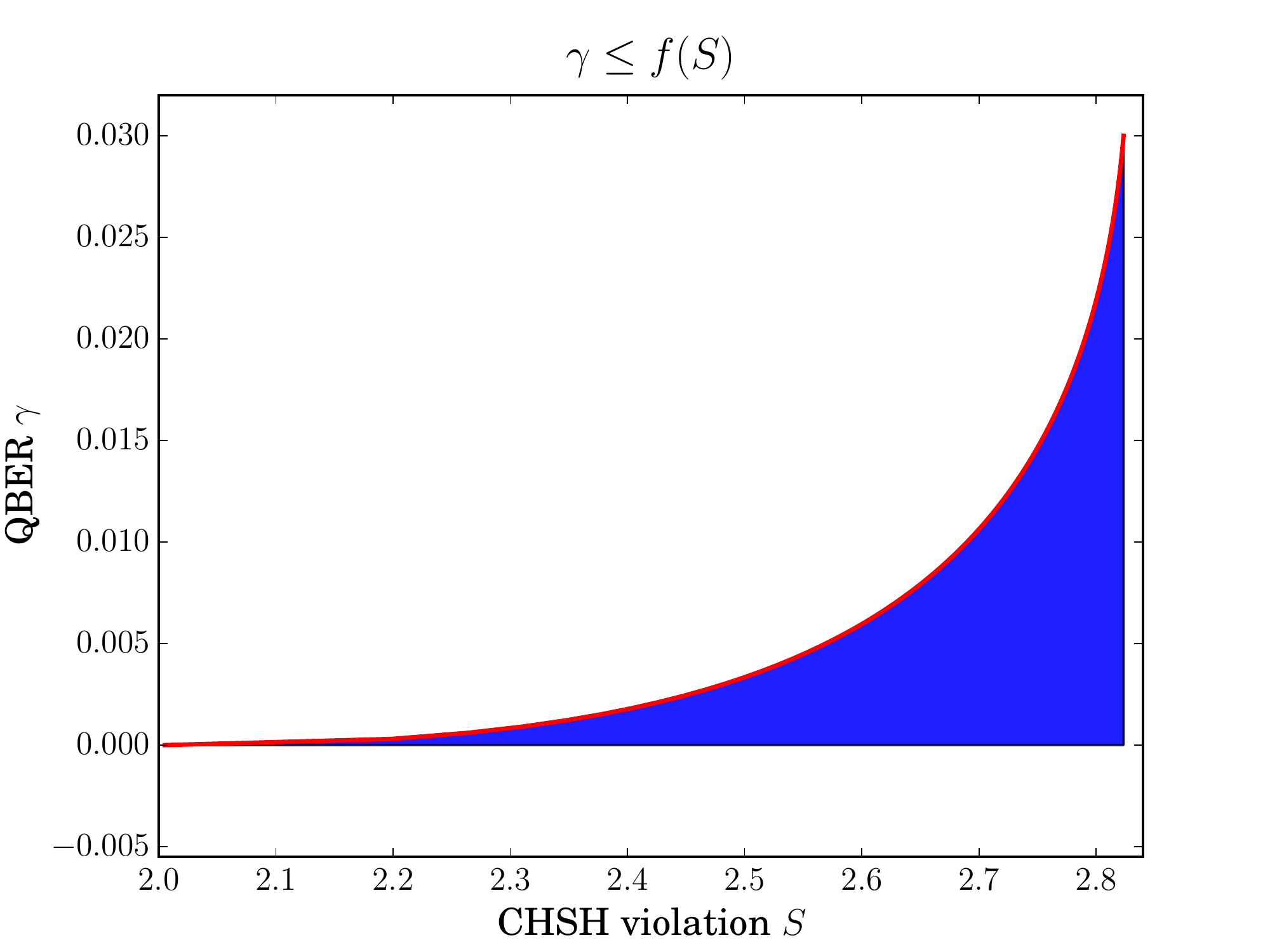}
    \caption{QBER $\gamma$ allowed in function of the CHSH violation $S$ obtained in the testing procedure when $n \rightarrow \infty$ and $d$ finite. The blue region is the secure region \emph{i.e.}~the region where the bound $B'(n,d,\zeta,\gamma)$ decays exponentially in $n$ for a fixed $d$.}
    \label{fig:noise_tolerance}
\end{figure}

Note that PV is still secure if we allow
$V_1$'s device to send the string $x$ to the prover after $V_1$ makes
the measurements on his state. $V_1$ just has to wait long enough
before measuring his state. Then dishonest provers cannot use this
information since there is a time constraint on their answers.

\section{Conclusion}

By proving that the probability of winning an imperfect guessing game decays exponentially with the length of the string which has to be guessed, we prove the security of Weak String Erasure and Position Verification even when the devices used to create or measure states are not trusted, i.e. in the device-independent scenario. This implies the device-independent security of any two-party cryptography protocols that can be made from WSE and classical communications, and certain position-based cryptographic tasks. Unfortunately, our proofs are not fully device-independent since we still assume that the honest party's devices (Alice's or $V_1$'s ones) are memoryless (the i.i.d.~assumption). Nevertheless, this is an improvement over previous known results, since until now bounds made under the memoryless devices were proportional to the dimension of the memory held by the adversary, whereas our bounds are proportional to the square root of this dimension. The improvement was made thanks to the development of new techniques, which directly deal with the adversary's quantum memory.

\begin{acknowledgements}
JR, JH and SW are funded by STW, NWO VIDI and ERC Starting Grant. LPT acknowledges the support of National Research Foundation, Prime Minister’s Office, Singapore and the Ministry of Education, Singapore under the Research Centres of Excellence programme. 
JK is supported by Ministry of Education, Singapore and the European Research Council (ERC Grant Agreement 337603).
\end{acknowledgements}

\appendix

\section{Technical Lemma}
In the proof of the Key Lemma to be presented below, we will need the following result. Similar results about norm of sums of operators have been obtained by Kittaneh~\cite{K97}, see also~\cite[Lemma 2]{TFKW13}.
\begin{lmm}
\label{lmm:normofsum}
If $A_1, A_2, \ldots, A_N$ are positive semidefinite operators, then
  \begin{align}
    \bigg\|{\sum_{i \in [N]} A_i}\bigg\| \leq \max_{j\in[N]} \sum_{i \in [N]}\, \left\|\sqrt{A_i}\sqrt{A_j}\right\|,
  \end{align}
where $[N]:=\{1, \ldots, N\}$.
\end{lmm}
\begin{proof}
Let $K$ be an $N\times N$ block matrix with entries $K_{ij}=\sqrt{A_i}\sqrt{A_j}$ and $L$ is an $N\times N$ matrix of entries $L_{ij}=\|\sqrt{A_i}\sqrt{A_j}\|$, we first show that
\begin{equation}
\label{eq:rewrite}
\bigg\|{\sum_{i \in [N]} A_i}\bigg\| = \|K\| \leq \|L\|\,.
\end{equation}
Defining $\tilde K := \sum_{j} \ket{j} \otimes \sqrt{A_{j}}$, a direct calculation reveals
\begin{align}%\simeq \sum_{j} A_{j}
\tilde K^{\dagger} \tilde K =  \sum_{j} A_{j}  \quad \textup{and} \quad \tilde K \tilde K^{\dagger} = \sum_{jk} \ketbra{j}{k} \otimes \sqrt{A_{j}} \sqrt{A_{k}} = K\,,
\end{align}
from which follows the first equality since the operator norm satisfies $\|\tilde K^{\dagger} \tilde K\| = \|\tilde K \tilde K^{\dagger} \|$. We are thus left with proving $\| \tilde K \tilde K^{\dagger} \| \leq \| L \|$ where now we rewrite $L$ in the following form
\begin{equation}
L = \sum_{jk} \ketbra{j}{k} \otimes \| \sqrt{A_{j}} \sqrt{A_{k}} \|\,.
\end{equation}
Since the operator norm of a positive semidefinite matrix corresponds to its largest eigenvalue, it suffices to prove that the largest eigenvalue of $\tilde K \tilde K^{\dagger}$ is not greater than the largest eigenvalue of $L$. Let $\ket{\alpha}$ be an eigenvector corresponding to the largest eigenvalue of $\tilde K \tilde K^{\dagger}$ and write it as
\begin{equation}
\ket{\alpha} = \sum_{j} \alpha_{j} \ket{j} \ket{e_{j}},
\end{equation}
where $\alpha_{j}$ are real and positive and $\ket{e_{j}}$ are arbitrary but normalised. Then
\begin{equation}
\| \tilde K \tilde K^{\dagger} \| = \bramatketq{\alpha}{\tilde K \tilde K^{\dagger}} = \sum_{jk} \alpha_{j} \alpha_{k} \bramatket{e_{j}}{\sqrt{A_{j}} \sqrt{A_{k}}}{e_{k}}.
\end{equation}
Now it suffices to prove that this can be upper bounded by $\bramatketq{\alpha'}{L}$ for
\begin{equation}
\ket{\alpha'} = \sum_{j} \alpha_{j} \ket{j},
\end{equation}
which implies 
\begin{equation}
\|K\|=\| \tilde K \tilde K^{\dagger} \| = \bramatketq{\alpha}{\tilde K \tilde K^{\dagger}} \leq \bramatketq{\alpha'}{L} \leq \|L\|.
\end{equation}
To show $\bramatketq{\alpha}{K} \leq \bramatketq{\alpha'}{L}$, we begin by rewriting $K$ as 
\begin{align}
    K&=\sum_{jk} \ketbra{j}{k} \otimes \sqrt{A_j} \sqrt{A_k}\\
    &=\sum_{j<k} \underbrace{\ketbra{j}{k} \otimes \sqrt{A_j} \sqrt{A_k}+ \ketbra{k}{j} \otimes \sqrt{A_k} \sqrt{A_j}}_{=:B_{jk}}+\sum_j \ketbra{j}{j} \otimes A_j
\end{align}
This form makes hermitian matrices $B_{jk}$ and $\ketbra{j}{j} \otimes A_j$ appear in the sums. $K$ is positive semidefinite so $\bramatket{\alpha}{K}{\alpha}=|\bramatket{\alpha}{K}{\alpha}|$ and,
\begin{align}
    |\bramatketq{\alpha}{K}|&=|\sum_{j \neq k} \alpha_j \alpha_k \bra{j}\bra{e_j} B_{jk}\ket{k}\ket{e_k}+ \sum_{j} \alpha_j^2 \underbrace{\bramatketq{e_j}{A_j}}_{\leq \|A_j\|}|\\
    & \leq\sum_{j \neq k} \alpha_j \alpha_k |\bra{j}\bra{e_j} B_{jk}\ket{k}\ket{e_k}|+ \sum_{j} \alpha_j^2  \|A_j\|. 
    \label{eq:bramaket_alpha}
\end{align}
Now by decomposing the vectors $\ket{j}\ket{e_j}=\sum_l \beta^j_l \ket{\beta_l}$ and $\ket{k}\ket{e_k}=\sum_m \beta^k_m \ket{\beta_m}$ in an eigenbasis of $B_{jk}$ noted $\{\ket{\beta_i}\}_i$ we get,
\begin{align}
    |\bra{j}\bra{e_j}B_{jk}\ket{k}\ket{e_k}|&=|\sum_{lm} \beta_l^{j*} \beta_m^k \bramatket{\beta_l}{B_{jk}}{\beta_m}|\\
    &=|\sum_{l} \beta_l^{j*} \beta_l^k \lambda_l|
\end{align} 
where $\{\lambda_l\}_l$ are the eigenvalues of $B_{jk}$. Using the triangle inequality we have,
\begin{align}
    |\sum_{l} \beta_l^{j*} \beta_l^k \lambda_l|&\leq \sum_{l} |\beta_l^{j*}| |\beta_l^k| |\lambda_l|\\
    &\leq \max_i |\lambda_{i}| \underbrace{\sum_{l} |\beta_l^{j*}| |\beta_l^k|}_{\leq 1}\\
    &\leq \max_i |\lambda_{i}| = \|B_{jk}\|. \label{eq:norm_B_jk}
\end{align}
It is easy to check that $\|B_{jk}\|=\|\sqrt{A_j}\sqrt{A_k}\|$. Using that and~\eqref{eq:norm_B_jk} in the inequality~\eqref{eq:bramaket_alpha} we have,
\begin{align}
    \|K\|=\bramatket{\alpha}{K}{\alpha}\leq \sum_{j k} \alpha_j \alpha_k \|\sqrt{A_j}\sqrt{A_k}\|=\bramatketq{\alpha'}{L}\leq \|L\|
\end{align}
which gives the desired inequality $\|K\|\leq \|L\|$.

Using H\"{o}lder's inequality (Lyapunov's inequalities) for induced $p$-norms, we have
\begin{equation}
    \|L\|=\|L\|_2^{\mathrm{I}}\leq \left(\|L\|_1^{\mathrm{I}} \cdot \|L\|_\infty^{\mathrm{I}}\right)^{1/2},
\end{equation}
where the norms on the right hand side are equal to the maximum absolute row or column sums
\begin{align}
    \|L\|_1^{\mathrm{I}}=\max_j \sum_i \|\sqrt{A_i}\sqrt{A_j}\|, \\
    \|L\|_\infty^{\mathrm{I}}=\max_i \sum_j \|\sqrt{A_i}\sqrt{A_j}\|.
\end{align}
The Lemma follows since these two norms are equal for hermitian matrices.
\end{proof}

\section{Proof of the Key Lemma}
The main content of this appendix is a detailed proof of the Key Lemma presented in the main text. Specifically, we prove a bound on the probability that Bob wins the game, only depending on a quantity $\zeta$ that Alice can estimate experimentally, Bob's memory size $d$, and $n$ which is the number of rounds played in the game. 

We split this proof into four steps. In Step 1 we analyse how Jordan's Lemma permits us to conveniently express the effective absolute anti-commutator of Alice's measurements. In Step 2 we derive a bound on the winning probability expressed in terms of what we call "operator overlap", then in Step 3 we bound this overlap by a simpler expression depending on the effective anti-commutator. We finish the proof in Step 4 by replacing, in the previous simple bound on the overlap, the effective anti-commutator by a quantity that Alice can estimate experimentally.

For the reader's convenience, we have included Table~\ref{tab:symbols} which explains the symbols used in the proof.

\subsection*{Step 1: Alice's i.i.d.~state-preparation and measurement device}
\label{sec:Jordan}
In this section we use Jordan's Lemma to rewrite Alice's measurement operators and the absolute effective anti-commutator.

We assume that the devices used by Alice to prepare and measure satisfy the i.i.d.~assumption, i.e.~the state produced in $n$ rounds is of the form $\rho_{AB}=\sigma_{AB}^{\otimes n}$ and the measurement corresponding to input $\theta\in\{0,1\}^n$ can be written as $\{P^\theta_x=\bigotimes_k P^{\theta_k}_{x_k}:x\in\{0,1\}^n\}$, where $\{P_0^0,P_1^0\}$ and $\{P_0^1,P_1^1\}$ are some unknown (but fixed) binary measurements. It is worth stressing that this implies that the reduced state on Alice is of product form, $\rho_{A} = \sigma_{A}^{\otimes n}$, \emph{regardless} of how Bob manipulates his subsystem. We make no assumptions on the dimensions of the system (except that they are finite dimensional).

%Recall that under the i.i.d.~assumption, the behaviour of $n$ interactions with the device can be described as follows: for the choice of , on the state $\rho_A=\sigma_A^{\otimes n}$. This is the only assumption we make on the joint state $\rho_{AB}$ between Alice and Bob, which at the moment have not yet appear in the picture. Also, we make no assumption on the size of Alice's Hilbert space on which these operators act, except that they are finite dimensional (which we do for simplicity).

By Naimark's dilation Theorem we can without loss of generality assume that the measurements $\{P_0^0,P_1^0\}$ and $\{P_0^1,P_1^1\}$ are projective,
%Because we work now with only two projective measurements,
which allows us to apply Jordan's Lemma ~\cite{NWZ09,Reg06}. % to characterize them completely. 
The projectors representing the first outcome $P_0^0$ and $P_0^1$ can be simultaneously decomposed into orthogonal projections of rank at most one, projecting on either one or two dimensional subspaces invariant under the action of both projectors (a subspace $W$ is invariant under the linear opearator $T$ if and only if $TW\subseteq W$). Moreover we only get two-dimensional blocks when the projectors have non-trivial overlap, i.e.~they are neither orthogonal nor identical:
\begin{equation}
P_0^0 = \sum_{j\in\mathcal{J}} P_{0|j}^0, \ \ P_0^1 = \sum_{j\in\mathcal{J}} P_{0|j}^1,\label{eq:Jordan}
\end{equation}
where $P_{0|j}^a$ ($a\in \{0,1\}$) are such that $\forall j,j'\ P_{0|j}^a P_{0|j'}^a=\delta_{jj'}P_{0|j}^a$ and are acting on a subspace of dimension one or two. Then we have,

\begin{equation}
P_1^0 = \sum_{j\in\mathcal{J}} (S_j-P_{0|j}^0), \ \ P_1^1 = \sum_{j\in\mathcal{J}} (S_j-P_{0|j}^1).
\end{equation} 
where $S_j$ is a projection which projects on the subspace where $P_{0|j}^a$ and $P_{1|j}^a$ act.
Note that the number of non-zero summands in~\eqref{eq:Jordan} is equal to the rank of $P_0^0$ or $P_0^1$ respectively.

A convenient basis of the Hilbert space can be chosen as follows. Each index $j\in\mathcal{J}$ corresponds to either a one or two dimensional subspace. For the two dimensional subspaces indexed by $j$ in $\mathcal{J}$ where \emph{both} $P_{0|j}^0$ and $P_{0|j}^1$ are nonzero, we pick the orthonormal eigenbasis of $P_{0|j}^0$, namely $|0_{|j}^{0}\rangle$ and $|1_{|j}^{0}\rangle$, as the basis for these subspaces
\begin{equation}
P_{0|j}^0 |0_{|j}^{0}\rangle = |0_{|j}^{0}\rangle, \ \ P_{0|j}^0 |1_{|j}^{0}\rangle = 0.
\end{equation}
Moreover, one can pick these basis vectors (by including a phase factor if necessary) such that the eigenstates $|0_{|j}^{1}\rangle$ and $|1_{|j}^{1}\rangle$ of $P_{0|j}^1$, which satisfy
\begin{equation}
P_{0|j}^1 |0_{|j}^{1}\rangle = |0_{|j}^{1}\rangle, \ \ P_{0|j}^1 |1_{|j}^{1}\rangle = 0,
\end{equation}
are related to those of $P_{0|j}^0$ by
\begin{equation}
|0_{|j}^{1}\rangle = \cos\beta_j |0_{|j}^{0}\rangle + \sin\beta_j |1_{|j}^{0}\rangle, \ \ |1_{|j}^{1}\rangle = \sin\beta_j |0_{|j}^{0}\rangle - \cos\beta_j |1_{|j}^{0}\rangle,
\label{eq:2dblock}
\end{equation}
for some angle $\beta_j\in[0,\pi/2]$. For the one dimensional subspaces (also indexed by $j\in\mathcal{J}$) we define $|0_{|j}^{0}\rangle$ being a unit vector, and $|1_{|j}^{0}\rangle$ is the null vector. Then we define $|0_{|j}^{1}\rangle$ and $|1_{|j}^{1}\rangle$ as previously but with  $\beta_j=0$ or $\beta_{j} = \pi/2$. In summary, since we have defined a basis for each $j\in\mathcal{J}$, taking the direct sum gives a basis for the whole Hilbert space. Any binary (projective) measurement device admits a characterization through the angles $\beta_j\in[0,\pi/2]$ for $j\in\mathcal{J}$ and this characterization turns out to be sufficient for our purposes.

\label{sec:avgeplus}
The previous block decomposition allows us to conveniently compute the effective absolute anticommutator
 defined as 
$\epsilon_+:=\frac{1}{2}\tr(|\{A_0,A_1\}|\sigma_A)$ where $A_\theta:=P_0^\theta-P_1^\theta$ for $\theta\in\{0,1\}$. The word "effective" means that $\epsilon_+$ depends not only on Alice's measurements, but also on the state on which the measurements act. Under Jordan's Lemma, the absolute anticommutator becomes
\begin{equation}
  |\{A_0,A_1\}|=|\sum_{j}\{A_{0_{|j}},A_{1_{|j}}\}|=\sum_{j} 2|\cos(2\beta_j)|S_j
\end{equation}
with $A_{\theta_{|j}}:=P_{0|j}^\theta-P_{1|j}^\theta$ and $S_j$ being the orthogonal projections defined above, where the absolute anticommutator of a two dimensional block $j$ is computed using~\eqref{eq:2dblock} and that of a one dimensional block follows from our definition of $\beta_j=0$ in Step 1. Let $p_{j} := \tr(S_j\sigma_A)$ be the probability of $\sigma_{A}$ being projected in the $j$-th block, then, the absolute effective commutator can be written as
%Defining the projection of $\sigma_A$ to each block $j$ of the Jordan's Lemma decomposition as $\sigma_{A_{|j}}:=S_j\sigma_AS_j/p_j$ with $p_j:=_j)$  have
\begin{equation}
  \epsilon_+=\sum_j p_j|\cos(2\beta_j)| = \sum_j p_j\epsilon_j,\label{eq:epsilonplus}
\end{equation}
where $\epsilon_j:=|\cos(2\beta_j)|$ is the absolute effective anticommutator of the block $j$. It is worth pointing out that for qubit observables there is no notion of "effectiveness", i.e.~the incompatibility is fixed by the observables and does not depend on the state.

Also, the previous decomposition of Alice measurements enables the $n$ run projectors to be block diagonalized as
\begin{equation}
P_x^\theta = \bigotimes_{k=1}^n P_{x_k}^{\theta_k} = \bigotimes_{k=1}^n \sum_{b_k\in\mathcal{J}} P_{x_k|b_k}^{\theta_k} = \sum_{b\in\mathcal{J}^n} P_{x|b}^{\theta},\label{eq:block_meas}
\end{equation}
where $P_{x|b}^\theta := \bigotimes_{k=1}^n P_{x_k|b_k}^{\theta_k}$, and $\mathcal{J}$ is the set of indices which label blocks and $\mathcal{J}^n:=\mathcal{J}^{\times n}$. We denote the set of projectors associated with this direct sum decomposition by $\{S_{b}\}_{b \in \mathcal{J}^n}$, where $S_b:=\bigotimes_{k=1}^nS_{b_k}$. For each $k$, we have $P_{x_k|b_k}^{\theta_k}P_{x'_k|b'_k}^{\theta_k}=\delta_{x_k,x'_k}\delta_{b_k,b'_k}P_{x_k|b_k}^{\theta_k}$, and $P_{x_k|b_k}^{\theta_k}$ orthogonal to $P_{x'_k|b'_k}^{\theta'_k}$ whenever  $b_k\neq b'_k$. The analysis of the guessing game will rest on these orthogonality relations and the set of angles $\beta_j$ defined above.

\begin{table}
\begin{center}
\begin{tabular}{| c | c | c |}
\hline
 \textbf{Variable} & \textbf{Range} & \textbf{Meaning} \\ \hline
 $n$ & $\mathbb{N}$ & total number of (measurement) runs \\  
 $k$ & $[n]$ & index of the run (subscript) or \\
 & & the classical information of Bob (superscript)\\  
 $\theta$ & $\{0,1\}^n$ & measurement string \\
 $\theta_k$ & $\{0,1\}$ & $k^{\textup{th}}$ measurement \\
 $x$ & $\{0,1\}^n$ & output string \\
 $x_k$ & $\{0,1\}$ & $k^{\textup{th}}$ output \\
 $j$ & $\mathcal{J}$ & index of Jordan's Lemma decomposition \\  
 $b$ & $\mathcal{J}^n$ & vector indexing block combination \\
 $b_k$ & $\mathcal{J}$ & $k^{\textup{th}}$ element of $b$ \\
 $p_{(\cdot)}$ & [0,1] & probability of $(\cdot)$ (depending on context)\\
 $\beta_{b_k}$ & $[0,\pi/2]$ & angle of Alice's binary measurement in block $b_k$ \\
 $\epsilon_{b_k}$ & [0,1] & absolute effective anticommutator in block $b_k$\\ \hline
\end{tabular}
\end{center}
\caption{Table of symbols.}
\label{tab:symbols}
\end{table}

\subsection*{Step 2: From guessing probability to "operator overlaps"}
\label{sec:intuition}

The goal of this section is to bound Bob's wining probability in term of the overlap $\Big\|\sqrt{\Pi^{\theta',k}_{b}}\sqrt{\Pi^{\theta,k}_{b}} \Big\|$. To be precise we show that,
\begin{lmm}
    Bob's wining probability $\lambda(n,d,\zeta)$ defined in equation \eqref{Def_lambda} is bounded as follow:
    \begin{equation}
        \lambda(n,d,\zeta) \leq  \max_{\left\{p_k,\rho_{AB'}^k\right\}}\max_{\{\mathcal{F}^{\theta,k}\}} \sum_{k,b}p_kp_{b|k} \max_{\theta'} 2^{-n} \sum_\theta   \left\|\sqrt{\Pi^{\theta',k}_{b}}\sqrt{\Pi^{\theta,k}_{b}} \right\| %\label{1_step_bound}.
    \end{equation}
\end{lmm}
\begin{proof}
Since we assume that the quantum memory of Bob is bounded he cannot store the entire register $B$ received from Alice. More specifically, according to the bounded storage model he must %any amount of quantum information $B$ exceeding his storage capacity must be preprocessed by Bob. For instance, he may decide to measure part of the system and store the other and, his strategy may depend on the set of possible measurements that Alice could perform. We represent this state of affairs
immediately input the register $B$ into an encoding map which outputs a quantum register $B'$ (whose dimension is bounded by $d$) and a classical register $K$ (of arbitrary size). %Bob now holds the qc-state $\rho_{B'K}$. 
The joint state between Alice and Bob is then a qqc-state $\rho_{AB'K}=\sum_k p_k\rho_{AB'}^k\otimes|k\rangle\langle k|$ whose marginal remains i.i.d.~$\rho_A=\tr_{B'K}(\rho_{AB'K})=\sigma_A^{\otimes n}$. Once Alice has measured her part of the system, Bob is told the choice of her measurements represented by $\theta \in \{0, 1\}^{n}$ and is asked to guess the string of outcomes. We take the success probability given by~\eqref{Def_lambda} and expand the classical register $K$ to obtain
\begin{align}
\lambda(n,d,\zeta) &=\max_{\substack{\rho_{AB'K} \\ \mathrm{qqc}}}
  \max_{\{\mathcal{F}^\theta\}} \tr\left(2^{-n}
  \sum_{\theta,x \in \{0,1\}^n} P_x^\theta
  \otimes F_x^{\theta} \quad \rho_{AB'K}\right)\\
&= \max_{\left\{p_k,\rho_{AB'}^k\right\}}\max_{\{\mathcal{F}^\theta\}}  \sum_kp_k\tr \left(\sum_{\theta,x}2^{-n} P^{\theta}_x\otimes F^{\theta}_x \rho_{AB'}^k\otimes\ketbra{k}{k}_K \right) \\
 &= \max_{\left\{p_k,\rho_{AB'}^k\right\}}\max_{\{\mathcal{F}^{\theta,k}\}}  \sum_kp_k\tr \left(\sum_{\theta,x}2^{-n} P^{\theta}_x\otimes F^{\theta,k}_x \rho_{AB'}^k \right), \label{eq:p_cheat_rewrite}
\end{align}
where $P_x^\theta$ are the measurement operators on Alice's side, and 
\begin{equation}
F^{\theta,k}_x:=\tr_K(F^\theta_x \mathds{1}_{B'}\otimes\ketbra{k}{k}_K)
\end{equation}
are $d$-dimensional measurement operators on Bob's side acting on $B'$, which depend both on his classical memory $k$ and the basis string $\theta$ received from Alice. The outer optimization is constrained to ensembles which yield the correct marginal on Alice's side, i.e.~$\tr_{B'}(\sum_kp_k\rho_{AB'}^k)=\sigma_A^{\otimes n}$. The inner maximization represents $|\Theta||K|$ independent maximizations each of which is over a POVM $\{F^{\theta,k}_x\}$. The rest of the proof will be concerned with upper bounding $\lambda(n,d,\zeta)$.

Inserting~\eqref{eq:block_meas} into \eqref{eq:p_cheat_rewrite} we get
\begin{equation}
%\lambda &= \max_{\{\mathcal{F}^\theta\}} \max_{\rho_{AB'}} \tr \left(2^{-n}\sum_{\theta}\Pi^\theta \rho_{AB'} \right) \\
\lambda(n,d,\zeta) = \max_{\left\{p_k,\rho_{AB'}^k\right\}}\max_{\{\mathcal{F}^{\theta,k}\}} \sum_kp_k\tr\left(\sum_{\theta,x}2^{-n} \sum_b P_{x|b}^{\theta}\otimes F_x^\theta \rho_{AB'}^k \right),
\end{equation}
Recall that $S_b$ represents the projection operator into the blocks indexed by $b\in\mathcal{J}^n$ . Define $\rho_{A_bB'}^k := (S_b\otimes\mathds{1}_{B'})\rho_{AB'}^k(S_b\otimes\mathds{1}_{B'})/p_{b|k}$ to be the normalized projections of $\rho_{A_bB'}^k$ into these various blocks with $p_{b|k}:=\tr\big( (S_b\otimes\mathds{1}_{B'}) \rho_{AB'}^k \big)$. Then
\begin{equation}
\lambda(n,d,\zeta) = \max_{\left\{p_k,\rho_{AB'}^k\right\}}\max_{\{\mathcal{F}^{\theta,k}\}} \sum_kp_k\sum_bp_{b|k} \tr\left(\sum_{\theta,x}2^{-n}   P_{x|b}^{\theta}\otimes F_x^{\theta,k} \rho_{A_bB'}^k \right), \label{eq:p_b|k}
\end{equation}
and for convenience let us denote $\Pi^{\theta,k}_b:=\sum_{x}  P_{x|b}^{\theta}\otimes F_x^{\theta,k}$ so that
 \begin{equation}
\lambda(n,d,\zeta) = \max_{\left\{p_k,\rho_{AB'}^k\right\}}\max_{\{\mathcal{F}^{\theta,k}\}} \sum_{k,b}p_kp_{b|k} \tr\left(\sum_\theta 2^{-n} \Pi^{\theta,k}_b \rho_{A_bB'}^k \right).
\end{equation}
%
% $\mathcal{H}_A^n=\otimes_k \oplus_{b_k} \mathcal{H}_{A_{b_k}} = \oplus_b \mathcal{H}_{A_b}^n$ where $\mathcal{H}_{A_b}^n = \otimes_k \mathcal{H}_{A_{b_k}}$.
%To derive a bound on $\lambda(n,d,\zeta)$ we first 
Bounding each of the trace terms by its operator norm yields
\begin{equation}
\lambda(n,d,\zeta) \leq \max_{\left\{p_k,\rho_{AB'}^k\right\}}\max_{\{\mathcal{F}^{\theta,k}\}} \sum_{k,b}p_kp_{b|k} \left\|\sum_\theta 2^{-n} \Pi^{\theta,k}_b\right\|.
\label{eq:lambda_first_bound}
\end{equation}
For each $b,k$ the corresponding operator norm can be bounded using Lemma~\ref{lmm:normofsum} as follows
\begin{equation}
\left\|\sum_\theta 2^{-n} \Pi^{\theta,k}_{b} \right\| \leq 2^{-n} \max_{\theta'}  \sum_\theta  \left\|\sqrt{\Pi^{\theta',k}_{b}}\sqrt{\Pi^{\theta,k}_{b}} \right\|% = \sum_\theta \left\|\sqrt{\Pi^{\theta'}_{b}}\sqrt{\Pi^\theta_{b}} \right\|,
\label{eq:bound_norm_sum}
\end{equation}
from which~\eqref{eq:lambda_first_bound} becomes
\begin{align}
%&\leq  \max_{\{\mathcal{F}^\theta\}} \sum_{b} \tilde p_{b} 2^{-n} \max_{\theta'} \sum_\theta \left\|\sqrt{\Pi^{\theta'}_{b}}\sqrt{\Pi^\theta_{b}} \right\|  \\
\lambda(n,d,\zeta) \leq  2^{-n} \max_{\left\{p_k,\rho_{AB'}^k\right\}}\max_{\{\mathcal{F}^{\theta,k}\}} \underbrace{\sum_{k,b}p_kp_{b|k} \max_{\theta'}  \sum_\theta   \left\|\sqrt{\Pi^{\theta',k}_{b}}\sqrt{\Pi^{\theta,k}_{b}} \right\|}_{=:\Lambda} \label{1_step_bound}.
\end{align}
\end{proof}

%Note that the attack of Bob, could for fixed $k$, result in Alice's state being entangled across different runs which in turn causes $p_{b|k}$ to depend on $b$ and $k$ simultaneously.

\subsection*{Step 3: Bound on operator overlaps}
\label{bounding-overlap}

The goal of this section is to find a bound on $\Big\|\sqrt{\Pi^{\theta',k}_{b}}\sqrt{\Pi^{\theta,k}_{b}} \Big\|$ which holds \textit{independenly of} $k$. 
The superscript $k$ of the operator $\Pi^{\theta',k}_{b}$ reminds us that Bob's measurement might depend on his classical information. Here, we derive a bound which only depends on the dimension of his quantum system i.e. independent of $k$. Therefore, we will from now omit the superscript $k$ which represented the classical information of Bob. The following Lemma is the key towards the main result.
\begin{lmm}
\label{lmm:normbnd}
For all $\theta', \theta \in \{0, 1\}^{n}$ and $b \in \mathcal{J}^{n}$, we have
\begin{equation}
\left\|\sqrt{\Pi^{\theta'}_{b}}\sqrt{\Pi^\theta_{b}} \right\| \leq \min\left\{1,\sqrt{d} \prod_{k=1}^n\left(\max\{\cos\beta_{b_{k}},\sin\beta_{b_{k}}\}\right)^{w_k} \right\},
\end{equation}
where $d$ is the dimension of Bob's quantum memory, and $w:=\theta'\oplus \theta \in \{0,1\}^n$
\end{lmm}
\begin{proof}
Let us begin by simplifying $\sqrt{\Pi^{\theta'}_{b}}\sqrt{\Pi^\theta_{b}}$ using the definition of $\Pi^\theta_b$ in Step 2 and orthogonality relations of $P^\theta_{x|b}$ in Step 1. Let $S=\{k\in[n]:\theta'_k=\theta_k\}$ and $T=\{k\in[n]:\theta'_k\neq\theta_k\}$ be the indices where the measurement choices agree and differ respectively. Then,
\begin{align}
\sqrt{\Pi_{b}^{\theta'}}\sqrt{\Pi^\theta_{b}} =& \sum_{x,y} P_{x|b}^{\theta'}P_{y|b}^{\theta} \otimes \sqrt{F_x^{\theta'}}\sqrt{F_y^\theta}\\
=& \sum_{x,y} \bigotimes_{k\in S} P_{x_k|b_k}^{\theta'_k}P_{y_k|b_k}^{\theta_k} \bigotimes_{k\in T} P_{x_k|b_k}^{\theta'_k}P_{y_k|b_k}^{\theta_k} \otimes \sqrt{F_x^{\theta'}}\sqrt{F_y^\theta}\\
=& \sum_{x,y} \underbrace{\bigotimes_{k\in S} \delta_{x_k,y_k}P_{y_k|b_k}^{\theta_k}}_{(\star)} \bigotimes_{k\in T} \left|x^{\theta'_k}_{k|b_k}\right>\left<x^{\theta'_k}_{k|b_k} \Big|y^{\theta_k}_{k|b_k}\right>\left<y^{\theta_k}_{k|b_k}\right| \otimes \sqrt{F_x^{\theta'}}\sqrt{F_y^\theta}
\end{align}
 where $\left|y^{\theta_k}_{k|b_k}\right>$  are the eigenvectors defined in Step 1. The notation $\left|y^{\theta_k}_{k|b_k}\right>$ should be read as the eigenvector representing the outcome $y_k\in\{0,1\}$ of the measurement $\theta_k\in\{0,1\}$ restricted to the block $b_k\in\mathcal{J}$.

The sum over $x=(x_1, ..., x_n) ,y=(y_1, ..., y_n)\in\{0,1\}^n$ can be split into a sum of variables with indices in $S$ which we will denote by $\tilde{x}, \tilde y$ and indices in $T$ denoted $x',y'$. The definition of $S$ implies that $\tilde{x}= \tilde y$. Let
\begin{align}
M_{\tilde{x}} := &\sum_{x',y'\in \{0,1\}^{|T|}} \bigotimes_{k\in T} \left|x^{\theta'_k}_{k|b_k}\right>\left<x^{\theta'_k}_{k|b_k} \Big|y^{\theta_k}_{k|b_k}\right>\left<y^{\theta_k}_{k|b_k}\right| \otimes \sqrt{F_{\tilde{x} x'}^{\theta'}}\sqrt{F_{\tilde{x} y'}^\theta},
\end{align}
where the two strings $x'$ and $x$ together form a bit string of length $n$ and give meaning to the notation $F_{\tilde{x} x'}^{\theta'}$ (same for $F_{\tilde{x} y'}^{\theta}$). Since the register labelled as $(\star)$ consists of orthogonal projectors, the desired operator norm is achieved by a maximization over $\tilde{x} \in \{0,1\}^{|S|}$. That is
\begin{equation}
\left\|\sqrt{\Pi_{b}^{\theta'}}\sqrt{\Pi^\theta_{b}}\right\| = \max_{\tilde x\in\{0,1\}^{|S|}} \left\| M_{\tilde{x}} \right\|
\end{equation}
In the following we derive an upper bound which does not depend on $\tilde{x}$. Since for $k \in T$ we have $\theta_{k} \neq \theta_{k}'$, we use Eq.~\eqref{eq:2dblock} to evaluate the inner product
\begin{equation}
\left<x^{\theta'_k}_{k|b_k} \Big|y^{\theta_k}_{k|b_k}\right> = (-1)^{x_{k} y_{k} } \cos(\beta_{b_k})^{\overline {x_k \oplus y_k}}
  \sin(\beta_{b_k})^{x_k \oplus y_k},
\end{equation}
where $\overline {x_k \oplus y_k} = 1- x_k \oplus y_k$.
Therefore,
%Using the basis chosen in Step 1 we can simplify this operator to
\begin{align}
M_{\tilde{x}} = \sum_{x',y'} (-1)^{x' \cdot \bar y'} \bigotimes_{k\in T} \cos(\beta_{b_k})^{\overline {x_k \oplus y_k}}
  \sin(\beta_{b_k})^{x_k \oplus y_k}   \left|x_{k|b_k}^{\theta'_k}\right>\left<y^{\theta_k}_{k|b_k}\right|\otimes \sqrt{F_{\tilde{x}x'}^{\theta'}}\sqrt{F_{\tilde{x}y'}^\theta}.
\end{align}
where $\forall x,y \in \{0,1\}^m,\ x.y:=\sum_k x_k y_k$

From $\|M_{\tilde{x}}\|=\sqrt{\|M_{\tilde{x}}M_{\tilde{x}}^\dagger\|}$ and the definition
\begin{equation}
\label{def:f}
f(\beta_{b},x',y',z'):= (-1)^{(x' \oplus z') \cdot \bar y} \prod_{k \in T} \cos(\beta_{b_k})^{\overline {x_k \oplus y_k}+ \overline {z_k \oplus y_k}} \sin(\beta_{b_k})^{x_k \oplus y_k+ z_k \oplus y_k},
\end{equation}
where $\beta_{b}$ is a vector of angles,
we simplify $M_{\tilde{x}}M_{\tilde{x}}^\dagger$ and get
\begin{align}
\|M_{\tilde{x}}\| &= \left\|\sum_{x',z'}\bigotimes_{k \in T}\left|x^{\theta'_k}_{k|b_k}\right>\left<z^{\theta'_k}_{k|b_k}\right| \otimes \underbrace{\sqrt{F^{\theta'}_{\tilde{x}x'}} \left(\sum_{y'}  f(\beta_{b},x',y',z') F_{\tilde{x}y'}^\theta\right) \sqrt{F^{\theta'}_{\tilde{x}z'}}}_{(**)} \right\|^{1/2}.
\end{align}
Bounding the register labelled as $(**)$ by its operator norm does not decrease the norm as mentioned in Lemma \ref{lmm:normofsum}. Hence we have
\begin{align}
\|M_{\tilde{x}}\| \leq & \left\|\sum_{x',z'}\bigotimes_{k \in T}\left|x^{\theta'_k}_{k|b_k}\right>\left<z^{\theta'_k}_{k|b_k}\right| \cdot \Big\|\sqrt{F^{\theta'}_{\tilde{x}x'}} \left(\sum_{y'}  f(\beta_{b},x',y',z') F_{\tilde{x}y'}^\theta\right) \sqrt{F^{\theta'}_{\tilde{x}z'}}\Big\| \right\|^{1/2}.
\end{align}
Bounding the outer operator norm with Schatten two norm ($\|\cdot\| \leq \|\cdot\|_2$) gives 
\begin{align}
\|M_{\tilde{x}}\| \leq & \left\|\sum_{x',z'}\bigotimes_{k \in T}\left|x^{\theta'_k}_{k|b_k}\right>\left<z^{\theta'_k}_{k|b_k}\right| \cdot \Big\|\sqrt{F^{\theta'}_{\tilde{x}x'}} \left(\sum_{y'}  f(\beta_{b},x',y',z') F_{\tilde{x}y'}^\theta\right) \sqrt{F^{\theta'}_{\tilde{x}z'}}\Big\| \right\|_2^{1/2}\\
=& \left(\sum_{x',z'} \left\|   \sqrt{F^{\theta'}_{\tilde{x}x'}} \left(\sum_{y'}  f(\beta_{b},x',y',z') F_{\tilde{x}y'}^\theta\right) \sqrt{F^{\theta'}_{\tilde{x}z'}} \right\|^2 \right)^{1/4}.
\end{align}
Using submultiplicativity of the operator norm we have
\begin{align}
\|M_{\tilde{x}}\| \leq&\left(\sum_{x',z'} \left\|\sqrt{F^{\theta'}_{\tilde{x}x'}}\right\|^2 \left\|\sum_{y'}  f(\beta_{b},x',y',z') F_{\tilde{x}y'}^\theta\ \right\|^2 \left\|\sqrt{F^{\theta'}_{\tilde{x}z'}}\right\|^2 \right)^{1/4} \\ \label{bound_overlap}
=&\left(\sum_{x',z'} \left\|F^{\theta'}_{\tilde{x}x'}\right\| \underbrace{\left\|\sum_{y'}  f(\beta_{b},x',y',z') F_{\tilde{x}y'}^\theta\ \right\|^2}_{(*)} \left\|F^{\theta'}_{\tilde{x}z'}\right\| \right)^{1/4}.
\end{align}
From the definition of $f(\beta_{b},x',y',z')$ in~\eqref{def:f} it is easy to see that
\begin{align}
| f(\beta_{b},x,y,z) | \leq \prod_{k \in T} \max\{\cos^2\beta_{b_{k}},\sin^2\beta_{b_{k}}\}.
\end{align}
Since this bound does not depend on $y'$ we can take it out of the sum
\begin{align}
(*) &\leq \prod_{k \in T} \max\{\cos^4\beta_{b_{k}},\sin^4\beta_{b_{k}}\} \left\|\sum_{y} F_{\tilde{x}y'}^\theta \right\|^2 \\
&\leq \prod_{k \in T} \max\{\cos^4\beta_{b_{k}},\sin^4\beta_{b_{k}}\}.
\end{align}
The latter inequality holds because $\sum_y F_{\tilde{x} y}^\theta \leq \mathds{1}$. Using this bound in \eqref{bound_overlap} gives us,
\begin{align}
  \left\|\sqrt{\Pi^{\theta'}_{b}}\sqrt{\Pi^\theta_{b}} \right\| &\leq \prod_{k \in T}\max\{\cos\beta_{b_{k}},\sin\beta_{b_{k}}\} \left[\sum_{x'z'} \|F^{\theta'}_{\tilde{x}x'} \| \|F^{\theta'}_{\tilde{x}z'}\|\right]^{1/4}\\
  &= \prod_{k \in T}\max\{\cos\beta_{b_{k}},\sin\beta_{b_{k}}\}\left[\sum_{x} \|F^{\theta'}_{\tilde{x}x'} \|\right]^{1/2} \\
  &\leq \sqrt{d} \prod_{k \in T}\max\{\cos\beta_{b_{k}},\sin\beta_{b_{k}}\}
\end{align}
where in the last inequality we use the observation that for all $\tilde{x}\in\{0,1\}^{n-t}$
\begin{align}
\sum_{x'} \| F^{\theta'}_{\tilde{x}x'} \| \leq \sum_x \tr(F^{\theta'}_{\tilde{x}x'})= \tr\left(\sum_x F^{\theta'}_{\tilde{x}x'}\right) \leq \tr(\mathds{1}_d) = d\,
\end{align}
since Bob's quantum memory is of dimension at most $d$.
Combining this bound with the trivial bound $\left\|\sqrt{\Pi^{\theta'}_{b}}\sqrt{\Pi^\theta_{b}} \right\|  \leq 1$ completes the proof.
\end{proof}

\begin{lmm}
\label{simp_nrm_bound}
For all $\theta', \theta \in \{0, 1\}^{n}$ and $b \in \mathcal{J}^{n}$, we have
\begin{equation}
\left\|\sqrt{\Pi^{\theta'}_{b}}\sqrt{\Pi^\theta_{b}} \right\| \leq \min\left\{1,\sqrt{d} \prod_{k=1}^n\left(\frac{1+\epsilon_{b_k}}{2}\right)^{w_k/2} \right\},
\end{equation}
where $d$ is the dimension of Bob's memory and $w:=\theta \oplus \theta' \in \{0,1\}^n$
\end{lmm}
\begin{proof}
Recall that for all $k\in[n]$ and $b\in \mathcal{J}^n$
\begin{equation}
    \epsilon_{b_k}:=|\cos2\beta_{b_k}|=\begin{cases}
    \cos2\beta_{b_k} &\text{ if } \beta_{b_k}\in[0,\pi/4] \\
    -\cos2\beta_{b_k} &\text{ if } \beta_{b_k}\in]\pi/4,\pi/2]
    \end{cases}
\end{equation}
\begin{itemize}
\item If $\beta_{b_k}\in[0,\pi/4]$ then $\cos\beta_{b_{k}}\geq\sin\beta_{b_{k}}$ and
    \begin{align}\max\{\cos\beta_{b_{k}},\sin\beta_{b_{k}}\}=\cos\beta_{b_{k}}= \sqrt{\frac{1+\cos 2 \beta_{b_k}}{2}}= \sqrt{\frac{1+\epsilon
    _{b_k}}{2}}. \end{align}
\item Similarly, if $\beta_{b_k}\in]\pi/4,\pi/2]$ then $\sin \beta_{b_k} \geq \cos \beta_{b_k}$ and
    \begin{align}\max\{\cos\beta_{b_{k}},\sin\beta_{b_{k}}\}=\sin\beta_{b_{k}}= \sqrt{\frac{1-\cos 2 \beta_{b_k}}{2}}= \sqrt{\frac{1+\epsilon
    _{b_k}}{2}}. \end{align}
\end{itemize}
Plugging
\begin{equation}
    \max\{\cos\beta_{b_{k}},\sin\beta_{b_{k}}\}=\sqrt{\frac{1+\epsilon_{b_k}}{2}},
\end{equation}
into Lemma~\ref{lmm:normbnd} completes the proof.
\end{proof}

\subsection*{Step 4: Completing the proof}
\label{sec:completing}

In this section we first we bound Bob's winning probability in terms of the absolute anti commutator $\epsilon_+$, and then bound it by a quantity that Alice can evaluate experimentally since it is a function of the Bell violation she estimates during the testing phase.

We are now in a position to relate the guessing probability with the average incompatibility $\epsilon_+$. That's to say we show that,
\begin{lmm} In~\eqref{1_step_bound} we bounded the winning probability in terms of $\Lambda$, which can be further bounded as follows
    \begin{align}
    %\Lambda(n,d,\zeta):=&\tr\left(2^{-n} \sum_{\theta,x \in \{0,1\}^n} P_x^\theta \otimes F_x^{\theta} \quad \rho_{AB'K}\right) \\
    %&\underset{\eqref{1_step_bound}}{\leq}
    \Lambda=2^{-n} \sum_{k,b}p_kp_{b|k} \max_{\theta'}  \sum_\theta   \left\|\sqrt{\Pi^{\theta',k}_{b}}\sqrt{\Pi^{\theta,k}_{b}} \right\| \leq 2^{-n}\sum_{w} \min \left(1,g\left(\vec{\epsilon}_+,w\right)\right),
    \end{align}
    where $w:=\theta'\oplus\theta\in\{0,1\}^n$, $\epsilon_+ = \sum_{j \in \mathcal{J}} p_j \epsilon_j$, $\vec \epsilon_+$ is the vector $(\epsilon_+,\epsilon_+,\ldots,\epsilon_+)$, and \[g(\vec a, w):=\sqrt{d}\prod_{k = 1}^{n} \left(\frac{1 + a_k}{2} \right)^{w_{k}/2},\]
    where $\vec a$ is a vector $(a_1,\ldots, a_k, \ldots, a_n)$.
    
    %Note that $\lambda(n,d,\zeta)=\max_{\substack{\rho_{AB'K} \\ \mathrm{qqc}}} \max_{\{\mathcal{F}^\theta\}} \Lambda(n,d,\zeta)$, which correspond to maximize $\Lambda(n,d,\zeta)$ over all Bob's strategy.
\end{lmm}

\begin{proof}
Define
\begin{align}
g(\epsilon_b,w) = \sqrt{d}\prod_{k = 1}^{n} \left(\frac{1 + \epsilon_{b_k}}{2} \right)^{w_{k}/2}, \, \, w:=\theta'\oplus\theta\in\{0,1\}^n. \label{eq:def_f}
\end{align}
When we apply Lemma~\ref{simp_nrm_bound} we get the bound
\begin{align}
    \max_{\theta'}2^{-n}\sum_\theta   \left\|\sqrt{\Pi^{\theta',k}_{b}}\sqrt{\Pi^{\theta,k}_{b}} \right\| \leq 2^{-n} \max_{\theta'}\sum_{\theta} \min\big(1, g(\epsilon_b,w) \big).
\end{align}
From \eqref{eq:def_f}, we observe that
\begin{align}
    \max_{\theta'} \sum_\theta \min\big(1, g(\epsilon_b,w) \big) = \sum_w \min\big(1, g(\epsilon_b,w) \big)
\end{align}
because the objective function to be maximized is independent of $\theta'$. 
The objective function in the optimization of~\eqref{1_step_bound} can be bounded as 
\begin{align}
\sum_{k,b} p_kp_{b|k} &\max_{\theta'} \sum_\theta 2^{-n} \left\|\sqrt{\Pi^{\theta',k}_{b}}\sqrt{\Pi^{\theta,k}_{b}}\right\|\\
&\leq \sum_{k,b} p_kp_{b|k} 2^{-n} 
\sum_{w} \min\big(1, g(\epsilon_b,w) \big),
\end{align}
where the inner expression is independent of $k$. This is the uniform bound we mentioned. Performing the sum over $k$ first gives
\begin{align}
    \sum_k p_kp_{b|k}  &\underset{\eqref{eq:p_b|k}}{=} \sum_k p_k \tr(S_b\otimes \mathds{1}_{B'}\rho^k_{AB'}S_b\otimes \mathds{1}_{B'}) \\&\underset{\phantom{\eqref{eq:p_b|k}}}{=} \tr(S_b\otimes \mathds{1}_{B'K}\ \rho_{AB'K}\ S_b\otimes \mathds{1}_{B'K})\\
    &\underset{\phantom{\eqref{eq:p_b|k}}}{=} \tr(S_b\sigma_A^{\otimes n}S_b) = \prod_{k=1}^n \tr(S_{b_k}\sigma_AS_{b_k}) = \prod_{k=1}^np_{b_k} =: p_b. \label{prod_p_b}
\end{align}
where $p_{b_k}$ for $b_k\in\mathcal{J}$ has been defined before~\eqref{eq:epsilonplus}.

Hence we see explicitly that while the attack of Bob may induce $p_{b|k}$ non-i.i.d.~for some $k$, on average he cannot influence Alice's local i.i.d.~state and therefore $p_{b}$ remains i.i.d.~

Swapping the order of summation over $b$ and $w$ and pulling the summation over $b$ inside the minimum (which can only increase the value) gives the upper bound
\begin{align}
2^{-n}&\sum_{b,w}p_{b} \min\big(1, g(\epsilon_b,w) \big) \\
&\leq 2^{-n}\sum_{w} \min \Big(1,\sum_{b} p_{b} g(\epsilon_b,w)\Big)
\end{align}
The sum inside the minimum,
\begin{align}
  \sum_b p_b g(\epsilon_b,w)= \sum_b p_b \sqrt{d}\prod_{k = 1}^{n} \left(\frac{1 + \epsilon_{b_k}}{2} \right)^{w_{k}/2},
\end{align}
is a product of sums because $p_b$ is a product (see eq.~\ref{prod_p_b}):
\begin{align}
  \sum_b p_b g(\epsilon_b,w)= \sqrt{d}\prod_{k = 1}^{n} \sum_{j \in \mathcal{J}} p_j \left(\frac{1 + \epsilon_{j}}{2} \right)^{w_{k}/2}.
\end{align}
Now each sum in the product can be bounded because of the concavity of the square root we get
\begin{align}
  \sum_b p_b g(\epsilon_b,w)&\leq \sqrt{d}\prod_{k = 1}^{n}  \left(\frac{1 + \epsilon_{+}}{2} \right)^{w_{k}/2}=g\left(\vec{\epsilon}_+,w\right)
\end{align}
hence we have
\begin{align}
&2^{-n}\sum_{b,w}p_{b} \min\big(1, g(\epsilon_b,w) \big)\\
&\leq 2^{-n}\sum_{w} \min \left(1,g\left(\vec{\epsilon}_+,w\right)\right) \label{eq:objective_bound}
\end{align}
where $\epsilon_+ = \sum_{j \in \mathcal{J}} p_j \epsilon_j$  and $\vec \epsilon_+$ is the vector $(\epsilon_+,\epsilon_+,\ldots,\epsilon_+)$
\end{proof}

The following Lemma forms the main result of this appendix. It bounds the winning probability  $\lambda(n,d,\zeta)$ by a function of $d$ and $\zeta$
\begin{lmm}
\label{lmm:main_tomo}
In the perfect guessing game where Alice's devices behave i.i.d.~and Bob has a quantum memory of dimension $d$, his winning probability is bounded by
\begin{align}
    \lambda(n,d,\zeta) \leq  2^{-n}\left[\sum_{k = 0}^{t} \binom{n}{k} + \sqrt{d} \sum_{k = t + 1}^{n} \binom{n}{k} \left(\frac{1+\zeta}{2} \right)^{k/2}\right],
\end{align}
where $t$ is the threshold defined as
\begin{align}
t := \left\lfloor - \log d \cdot \left[ \log \left(\frac{1+\zeta}{2}\right) \right]^{-1} \right\rfloor
\end{align}
and $\zeta:=\frac{S}{4}\sqrt{8-S^2}$  with $S$ being the \textup{CHSH} violation as defined in Lemma \ref{Bound_Eps_+}.
\end{lmm}
\begin{proof}
Combine~\eqref{eq:objective_bound} with~\eqref{eq:def_f} and~\eqref{1_step_bound} and note that the maximizations over all strategies of Bob drop out because we have bounded the winning probability of an arbitrary strategy. Therefore, we obtain
\begin{equation}
\lambda(n,d,\zeta) \leq 2^{-n}\sum_w \min\left\{1,\sqrt{d}\prod_{k = 1}^{n} \left(\frac{1 + \epsilon_+}{2} \right)^{w_{k}/2}\right\}.
\end{equation}
Using Lemma \ref{Bound_Eps_+} we have $\epsilon_+\leq \zeta$ and then
\begin{equation}
\lambda(n,d,\zeta) \leq 2^{-n}\sum_w \min\left\{1,\sqrt{d}\prod_{k = 1}^{n} \left(\frac{1 + \zeta}{2} \right)^{w_{k}/2}\right\}.
\end{equation}
Since the right-hand side depends only on the Hamming weight of $w\in\{0,1\}^n$ it is easy to perform the minimization explicitly, which yields
%
%Splitting the sum over sequences $w$ such that $g(\vec{\epsilon}_+,w)\leq1$ and its complement, and observe that only the Hamming weight of  matters, the right hand side becomes
\begin{align}
\lambda(n,d,\zeta) \leq 2^{-n} \left[\sum_{k = 0}^{t} \binom{n}{k} + \sqrt{d} \sum_{k = t + 1}^{n} \binom{n}{k} \left(\frac{1+\zeta}{2} \right)^{k/2}\right],
\end{align}
 where $t$ is the threshold defined in the Lemma.
\end{proof}

\end{document}